\documentclass[smallextended,envcountsame]{svjour3}
\smartqed 

\usepackage{microtype}

\usepackage{amssymb}
\usepackage{amsmath}
\usepackage{caption}
\usepackage[noadjust]{cite}
\usepackage{graphicx}
\usepackage{mdwlist}
\usepackage{subfig}
\usepackage{ulem}
\usepackage{tikz}
\usepackage{todonotes}
\usepackage{marginnote}
\usepackage{xspace}
\usepackage[pdfborder={0 0 0}]{hyperref}

\usetikzlibrary{positioning,fit,arrows,decorations.markings}
\usetikzlibrary{matrix}
\tikzset{
to/.style={->},
To/.style={->,postaction={decorate},
 decoration={markings,
 mark=at position .5 with
 \node[transform shape]{$\large\mathstrut\raisebox{0.02em}%
 {$\smallparallel$}$};
}},
TO/.style={->,postaction={decorate},
 decoration={markings,
 mark=at position .5 with
 \node[transform shape]{$\mathstrut\smalldevel$};
}},
tos/.style={->>},
hide/.style={color=white},
}

\newcommand{\CR}{\textnormal{CR}\xspace}
\newcommand{\PRE}{\textnormal{pre}\xspace}
\newcommand{\SN}{\textnormal{sn}\xspace}
\renewcommand{\SN}{\textnormal{rt}\xspace}
\newcommand{\RL}{\textnormal{rl}\xspace}
\newcommand{\seq}[2][n]{{#2_1},\dots,{#2_{#1}}}
\newcommand{\ST}[1]{{{\ensuremath\star\textnormal{(}{#1}%
\textnormal{)}}}}
\newcommand{\STD}[1]{{\ensuremath{#1^\star_\dagger}}}
\newcommand{\STG}[1]{{\ensuremath{#1^\star_>}}}
\newcommand{\STE}[1]{{\ensuremath{#1^\star_=}}}
\newcommand{\STC}[1]{{\ensuremath{#1^\star_\gtrsim}}}
\newcommand{\sstar}{{\ensuremath%
 {\makebox[0mm][l]{\raisebox{-0.45ex}{$\star$}}}%
 {\raisebox{0.65ex}{$\star$}}%
}}
\newcommand{\STT}[1]{\sstar \textnormal{(}{#1}\textnormal{)}}
\newcommand{\STTG}[1]{{\ensuremath{#1^{\star\star}_>}}}
\newcommand{\STTE}[1]{{\ensuremath{#1^{\star\star}_=}}}

\newcommand{\red}{{\triangle}}

\DeclareSymbolFont{letters}{OML}{cmbboard}{m}{it}

\newcommand{\XP}{(parallel)\xspace}
\newcommand{\XO}{(critical overlap)\xspace}
\newcommand{\XV}{(variable overlap)\xspace}
\newcommand{\XVL}{(variable-linear)\xspace}
\newcommand{\XVLL}{(variable-left-linear)\xspace}

\newcommand{\HOLE}{\Box}

\newcommand{\labgt}{>}
\newcommand{\labge}{\geqslant}
\newcommand{\lablt}{<}
\newcommand{\lable}{\leqslant}
\newcommand{\labvar}{\circ}
\newcommand{\labany}{\circ}
\newlength{\len}

\newcommand{\from}{\leftarrow}

\newcommand{\smallparallel}{%
 \def\next##1##2{\raise.07ex\hbox{$##1{\shortmid}\mkern-1.5mu%
{\shortmid}$}}%
 \mathpalette\next{}%
}
\newcommand{\smalldevel}{%
 \def\next##1##2{\raise.07ex\hbox{$##1\raisebox{0.02em}%
 {\scalebox{.8}{$\circ$}}$}}\mathpalette\next{}%
}
\newcommand{\overlayrel}[4]{%
 \mathrel{%
 \def\next##1##2{%
 \setbox0=\hbox{$##1#3$}%
 \dimen0=\wd0%
 \hbox to0pt{\box0}%
 \hbox to\dimen0{$##1\hfil\mkern#1mu#4\mkern#2mu\hfil$}}%
 \mathpalette\next{}}%
}
\newcommand{\prightarrow}{\overlayrel02\rightarrow\smallparallel}
\newcommand{\pleftarrow}{\overlayrel20\leftarrow\smallparallel}

\newcommand{\xprightarrow}[2][]{\overlayrel02{\xrightarrow[#1]{#2}}%
\smallparallel}
\newcommand{\xpleftarrow}[2][]{\overlayrel20{\xleftarrow[#1]{#2}}%
\smallparallel}
\newcommand{\orightarrow}{\overlayrel02\rightarrow\smalldevel}
\newcommand{\oleftarrow}{\overlayrel20\leftarrow\smalldevel}
\newcommand{\xorightarrow}[2][]{\overlayrel02{\xrightarrow[#1]{#2}}%
\smalldevel}
\newcommand{\xoleftarrow}[2][]{\overlayrel20{\xleftarrow[#1]{#2}}%
\smalldevel}

\newcommand\Nat{\ensuremath{\mathbb{N}}\xspace}
\newcommand\SIG[1]{\ensuremath{\mathcal{#1}}}
\newcommand\VAR[1]{\ensuremath{\mathcal{#1}}}
\newcommand\FF{\SIG{F}}
\newcommand\GG{\SIG{G}}
\newcommand\VV{\VAR{V}}

\newcommand\TERMS[2]{\mbox{\ensuremath{\mathcal{T}%
(\SIG{#1},\VAR{#2})}}\xspace}
\newcommand\FVTERMS{\TERMS{F}{V}}

\newcommand\STEPS[1]{\ensuremath{\Lambda_{#1}\xspace}}
\newcommand\RR{\ensuremath{\mathcal{R}}\xspace}
\renewcommand\SS{\ensuremath{\mathcal{S}}\xspace}
\newcommand\RRd{{{\ensuremath{\mathcal{R}_\mathsf{d}}}}\xspace}
\newcommand\RRnd{{\ensuremath{\mathcal{R}_\mathsf{nd}}}\xspace}

\newcommand\REL[2]{{\ensuremath{#1 \kern0em/\kern0em #2}}\xspace}
\newcommand\RS{\REL{\RR}{\SS}}
\newcommand\SR{\REL{\SS}{\RR}}
\newcommand\RRdnd{\REL{\RRd}{\RRnd}}

\newcommand{\sep}{\hspace{-0.28em}}
\newcommand{\cp}{\mathrel{\from\sep\rtimes\sep\to}}
\newcommand{\cps}[1]{\mathrel{\from\sep{#1}\sep\to}}
\renewcommand{\cps}[1]{\from #1 \to}
\newcommand{\pcp}{\mathrel{\mparr\sep\rtimes\sep\to}}

\newlength{\rtimesl}
\newcommand{\overlay}%
{\mathrel{\from\sep{\ltimes\hspace{-\rtimesl}\rtimes}\sep\to}}
\newcommand\PCPS{\ensuremath{\mathsf{PCPS}}}
\newcommand\CPS{\ensuremath{\mathsf{CPS}}}
\newcommand\CDS{\ensuremath{\mathsf{CDS}}}

\newcommand\NAT{\mathbb{N}}
\newcommand\Var{\mathcal{V}\mathsf{ar}}
\newcommand\Pos{{\mathcal{P}\mathsf{os}}}
\newcommand\FPos{{\Pos_\mathcal{F}}}
\newcommand\VPos{{\Pos_\mathcal{V}}}

\newcommand\BM{\hspace{-0.3ex}\begin{pmatrix}}
\newcommand\EM{\end{pmatrix}\hspace{-0.3ex}}

\newcommand\m[1]{\mathsf{#1}}
\newcommand\x[1]{\mathcal{#1}}

\newcommand\ACP{\ensuremath{\mathsf{ACP}}\xspace}
\newcommand\CSI{\ensuremath{\mathsf{CSI}}\xspace}
\newcommand\SAIGAWA{\ensuremath{\mathsf{saigawa}}\xspace}
\newcommand\TTTT{%
 \ensuremath{\mathsf{T\kern-0.2em\raisebox{-0.3em}%
 {$\mathsf{T}$}\kern-0.2emT\kern-0.2em%
 \raisebox{-0.3em}$\mathsf{2}$}}\xspace%
}

\makeatletter
\newcommand{\r@rrow}[3]{%
 \newcommand{#1}[2][]{%
  \def\next{#2\@ifempty{##1}{}{_{##1}}\@ifempty{##2}{}{^{##2}}}%
  \mathchoice{#3[##1]{##2}}{\next}{\next}{\next}%
 }%
}
\newcommand{\l@rrow}[3]{%
 \newcommand{#1}[2][]{%
  \def\next####1{%
   \setbox0=\hbox{$####1\vphantom{#2}\@ifempty{##1}{}{_{}}%
    \@ifempty{##2}{}{^{##2}}$}%
   \setbox1=\hbox{$####1\vphantom{#2}\@ifempty{##1}{}{_{##1}}%
    \@ifempty{##2}{}{^{}}$}%
   \setbox2=\vbox{\hbox to\wd0{}\hbox to\wd1{}}%
   \mathrel{\hskip\wd2\hskip-\wd0\box0\hskip-\wd1\box1{#2}}%
  }%
  \mathchoice{#3[##1]{##2}}{\next\textstyle}{\next\scriptstyle}%
 {\next\scriptscriptstyle}%
 }%
}
\r@rrow{\xr}{\rightarrow}{\xrightarrow}
\l@rrow{\xl}{\leftarrow}{\xleftarrow}
\r@rrow{\xR}{\prightarrow}{\xprightarrow}
\l@rrow{\xL}{\pleftarrow}{\xpleftarrow}
\r@rrow{\XR}{\orightarrow}{\xorightarrow}
\l@rrow{\XL}{\oleftarrow}{\xoleftarrow}
\r@rrow{\xrw}{\mathrel{\widetilde\Rightarrow}}{\xrightarrow}
\r@rrow{\xrb}{\Rightarrow}{\xrightarrow}
\l@rrow{\xlw}{\mathrel{\widetilde\leftarrow}}{\xleftarrow}
\makeatother

\let\parfrom\pleftarrow
\let\parto\prightarrow
\newcommand{\lab}{\ell}
\newcommand{\plab}{\lab^\parallel}

\def\pcp{\mathrel{{\parfrom}{\rtimes}{\to}}}

\def\pcpk#1#2#3{\xpleftarrow[\{#1\}]{} #2 \xrightarrow[\{#3\}]{}}

\newcommand{\turnleft}[1]{%
 \def\next##1##2{\raisebox{-.09em}{\rotatebox{90}{$##1#1$}}}%
 \mathpalette\next{}%
}
\newcommand{\Vee}{\turnleft\lablt}
\newcommand{\Veq}{\turnleft\lable}

\newcommand{\mul}{\textsf{mul}}

\newcommand{\?}{\mathbin{?}}

\begin{document}
\title{Labelings for Decreasing Diagrams\footnote{This research is
supported by FWF (Austrian Science Fund) project P22467.}}
\titlerunning{Labelings for Decreasing Diagrams}

\author{Harald Zankl, Bertram Felgenhauer, and Aart Middeldorp}
\authorrunning{H.~Zankl et al.}

\institute{
H.\ Zankl, B.\ Felgenhauer, A.\ Middeldorp \at
Institute of Computer Science, University of Innsbruck,
6020 Innsbruck, Austria
\and
H.\ Zankl \at
\email{harald.zankl@uibk.ac.at}
\and
B.\ Felgenhauer \at
\email{bertram.felgenhauer@uibk.ac.at}
\and
A.\ Middeldorp \at
\email{aart.middeldorp@uibk.ac.at}
}

\date{Received: date / Accepted: date}
 
\maketitle

\begin{abstract}
This article is concerned with automating the decreasing diagrams
technique of van Oostrom for establishing confluence of term rewrite
systems. We study abstract criteria that allow to lexicographically
combine labelings to show local diagrams decreasing. This approach
has two immediate benefits. First, it allows to use labelings for
linear rewrite systems also for left-linear ones, provided some mild
conditions are satisfied. Second, it admits an incremental method for
proving confluence which subsumes recent developments in automating
decreasing diagrams. The techniques proposed in the article have been
implemented and experimental results demonstrate how, e.g., the rule
labeling benefits from our contributions.
\end{abstract}

\section{Introduction}

Confluence is an important property of rewrite systems since it
ensures unique normal forms. It is decidable in the presence of
termination~\cite{KB70} and implied by orthogonality~\cite{R73} or restricted
joinability conditions on the critical pairs~\cite{H80,T88,vO97,OO97,O98}.
Recently, there is a renewed interest in confluence research, 
with a strong emphasis on automation. As one application we 
mention~\cite{SZKO12}, where automated confluence tools are employed
for proving soundness of abstract forms of reduction in solving the
typing problem.

The decreasing diagrams technique of van Oostrom~\cite{vO94} is a
complete method for showing confluence of countable abstract rewrite
systems.
The main idea of the approach is to show confluence by establishing
local confluence under the side condition that rewrite steps of the
joining sequences must \emph{decrease} with respect to some
well-founded order. For term rewrite systems however, the main
problem for automation of decreasing diagrams is that in general
infinitely many local peaks must be considered. To reduce this
problem to a finite set of local peaks one can
label rewrite steps with functions that satisfy special properties.
In~\cite{vO08a} van Oostrom presented the rule labeling that
allows to conclude confluence of \emph{linear} rewrite systems by
checking decreasingness of the critical peaks (those emerging from
critical overlaps). The rule labeling has been
implemented by Aoto~\cite{A10} and Hirokawa and
Middeldorp~\cite{HM11}. Already in~\cite{vO08a}
van Oostrom presented constraints that allow to apply the rule
labeling to \emph{left-linear} systems. This approach has been
implemented and extended by Aoto~\cite{A10}. 
Our framework subsumes the above ideas.

The contributions of this article comprise the extraction of abstract
constraints on a labeling such that for a (left-)linear rewrite system
decreasingness of the (parallel) critical peaks ensures confluence.
We show that the rule labeling adheres to our constraints and present
additional labeling functions. Furthermore
such labeling functions can be combined lexicographically to obtain new
labeling functions satisfying our constraints. This approach
allows the formulation of an abstract criterion that makes
virtually every labeling function for linear rewrite systems also
applicable to left-linear systems. 
Consequently, confluence of the TRS in Example~\ref{EX:OO03} can be
established automatically, e.g., by the rule labeling,
while current approaches based on the decreasing diagrams
technique~\cite{A10,HM11} as well as
other confluence criteria like Knuth and Bendix' criterion
or orthogonality (and its refinements) fail.

\begin{example}
\label{EX:OO03}
Consider the TRS $\RR$ (Cops~\#60)\footnote{%
COnfluence ProblemS, see
\url{http://coco.nue.riec.tohoku.ac.jp/problems/}.}
consisting of the rules
\begin{alignat*}{4}
1\colon  &~&x + (y + z)      &\to (x + y) + z &\qquad
6\colon  &~&x \times y       &\to y \times x
\\
2\colon  &&(x + y) + z       &\to x + (y + z) &\qquad
7\colon  &&\m{s}(x) + y      &\to x + \m{s}(y)
\\
3\colon  &&\m{sq}(x)         &\to x \times x &\qquad
8\colon  &&x + \m{s}(y)      &\to \m{s}(x) + y
\\
4\colon  &&\m{sq}(\m{s}(x))  &\to (x \times x) + \m{s}(x + x) &\qquad
9\colon  &&x \times \m{s}(y) &\to x + (x \times y)
\\
5\colon  &&x + y             &\to y + x &\qquad
10\colon &&\m{s}(x) \times y &\to (x \times y) + y
\end{alignat*}
This system is locally confluent since all its 34 critical pairs are
joinable.
\end{example}

The remainder of this article is organized as follows. After recalling
preliminaries in Section~\ref{PRE:main} we present constraints 
(on a labeling) such that decreasingness of the critical peaks ensures
confluence for \mbox{(left-)}linear rewrite systems in
Section~\ref{LAB:one:main}.
Three of these constraints are based on relative termination while
the fourth employs persistence.
We focus on parallel rewriting in Section~\ref{LAB:par:main}.
The merits of these approaches are assessed in Section~\ref{ASS:main}
by discussing the relationship to the recent literature.
Implementation issues are addressed in Section~\ref{IMP:main} before 
Section~\ref{EXP:main} gives an empirical evaluation of our results.
Section~\ref{CON:main} concludes.

This article is an updated and extended version of~\cite{ZFM11},
which presents the first incremental approach for labeling decreasing
diagrams. Besides a number of small
improvements, the article contains three new major
contributions:
\begin{itemize}
\item
Section~\ref{LAB:ll:red}, presenting a new labeling measuring the
contracted redex,
\item
Section~\ref{LAB:ll:per}, which uses persistence to enhance
the applicability of L-labelings
for left-linear systems,
\item
Section~\ref{LAB:par:main}, which studies parallel rewriting to
make any weak LL-labeling applicable to showing confluence of
left-linear systems without additional (relative termination)
constraints.
\end{itemize}
The latter generalizes and incorporates recent findings 
from~\cite{F13}, which studies the rule labeling for parallel
rewriting.

\section{Preliminaries}
\label{PRE:main}

We assume familiarity with term rewriting \cite{BN98,TeReSe}.

Let $\FF$ be a signature and let $\VV$ be a set of variables disjoint
from $\FF$. By $\FVTERMS$ we denote the set of terms over $\FF$ and
$\VV$. The expression $|t|_x$ indicates how often variable $x$ occurs
in term $t$.
Positions are strings of natural numbers, i.e., elements of
$\NAT_+^*$. The set of positions of a term $t$ is defined as
$\Pos(t) = \{ \epsilon \}$ if $t$ is a variable and as
$\Pos(t) = \{ \epsilon \} \cup \{ iq \mid
\text{$1 \leqslant i \leqslant n$ and
$q \in \Pos(t_i)$} \}$
if $t = f(\seq{t})$. 
We write $p \leqslant q$ if $q = pp'$ for some position $p'$,
in which case $q \backslash p$ is defined to be $p'$.
Furthermore $p < q$ if $p \leqslant q$ and $p \neq q$.
Finally, $p \parallel q$ if neither $p \leqslant q$ nor $q < p$.
Positions are used to address subterm occurrences.
The subterm of $t$ at position $p \in \Pos(t)$ is defined as
$t|_p = t$ if $p = \epsilon$ and as
$t|_p = t_i|_q$ if $p = iq$. We write 
$u \trianglelefteq t$ if $u$ is a subterm of $t$ and
$s[t]_p$ for the result of
replacing $s|_p$ with $t$ in $s$. 
The set of function symbol positions $\FPos(t)$ is
$\{ p \in \Pos(t) \mid t|_p \notin \VV \}$ and
$\VPos(t) = \Pos(t) \setminus \FPos(t)$.
The set of variables occurring in a term $t$ is 
denoted by $\Var(t)$. We let $t|_P = \{ t|_p \mid p \in P \}$
if $t$ is a term and $P$ a set of positions.

A rewrite rule is a pair of terms $(l,r)$, written $l \to r$,
such that $l$ is not a variable and all variables in $r$ are contained
in $l$. A rewrite rule $l \to r$ is duplicating if $|l|_x < |r|_x$ for
some $x \in \VV$. A term rewrite system (TRS) is a 
signature together with a finite set of rewrite rules
over this signature. In the sequel signatures are implicit.
By $\RRd$ and $\RRnd$ we denote the duplicating and non-duplicating
rules of a TRS $\RR$, respectively.
A rewrite relation is a binary relation on terms that is
closed under contexts and substitutions. For a TRS $\RR$ we define
$\to_\RR$ to be the smallest rewrite relation that contains $\RR$.
As usual $\to^=$, $\to^+$, and $\to^*$ denotes the reflexive,
transitive, and reflexive and transitive closure of $\to$, respectively.

A relative TRS $\RS$ is a pair of TRSs $\RR$ and $\SS$ with the
induced rewrite relation 
${\to_\RS} = {\to_{\SS}^* \cdot \to_{\RR}^{} \cdot \to_{\SS}^*}$.
Sometimes we identify a TRS $\RR$ with the relative
TRS $\REL{\RR}{\varnothing}$ and vice versa.
A TRS $\RR$ (relative TRS $\RS$) is terminating if $\to_\RR$
($\to_\RS$) is well-founded.
Two relations $\geqslant$ and $>$ are called compatible if 
${\geqslant} \cdot {>} \cdot {\geqslant} \subseteq {>}$.
A monotone reduction pair $(\geqslant,>)$ consists of a
preorder $\geqslant$ and a well-founded order $>$ such that
$\geqslant$ and $>$ are compatible and closed under contexts and
substitutions. 
A reduction pair $(\geqslant,>)$ is called simple if
${f(s_1,\ldots,s_n) \geqslant s_i}$ for all $1\leqslant i\leqslant n$.
We recall how to prove relative
termination incrementally according to Geser~\cite{G90}.

\begin{theorem}
\label{THM:rt}
A relative TRS $\RS$ is terminating if $\RR = \varnothing$ or there 
exists a monotone reduction pair $(\geqslant,>)$ such that
$\RR \cup \SS \subseteq {\geqslant}$ and
$\REL{(\RR \setminus {>})}{(\SS \setminus {>})}$ is terminating.
\qed
\end{theorem}

A critical overlap
$(l_1 \to r_1,p,l_2 \to r_2)_\mu$ of a TRS $\RR$ consists
of variants $l_1 \to r_1$ and $l_2 \to r_2$ of rewrite rules of $\RR$
without common variables, a position $p \in \FPos(l_2)$, and a most
general unifier $\mu$ of $l_1$ and $l_2|_p$. If $p = \epsilon$ then we
require that $l_1\to r_1$ and $l_2 \to r_2$ are not variants.
From a critical overlap
$(l_1 \to r_1,p,l_2 \to r_2)_\mu$ we obtain a 
critical peak $l_2\mu[r_1\mu]_p \cps{l_2\mu} r_2\mu$ and a
critical pair $l_2\mu[r_1\mu]_p \cp r_2\mu$.

If $l \to r \in \RR$ and $p$ is
a position, we call the pair $\pi = \langle p, l \to r \rangle$ a
redex pattern, and write $l_\pi$, $r_\pi$, $p_\pi$ for its left-hand
side, right-hand side, and position, respectively. We write
$\xr{\pi}$ (or $\xr{p_\pi,l_\pi\to r_\pi}$) for a rewrite step at
position $p_\pi$ using the rule
$l_\pi \to r_\pi$. A redex pattern
$\pi$ matches a term $t$ if $t|_{p_\pi}$ is an instance of $l_\pi$.
If $\pi$ matches $t$, there is a unique reduct $t^\pi$ with
$t \xr{\pi} t^\pi$.

Let $\pi_1$ and $\pi_2$ be redex patterns that match a common term.
They are called parallel
($\pi_1 \parallel \pi_2$) if $p_{\pi_1} \parallel p_{\pi_2}$.
If $p_{\pi_2} \leqslant p_{\pi_1}$ and
$p_{\pi_1} \!\!\setminus p_{\pi_2} \in \Pos_\FF(l_{\pi_2})$
then $\pi_1$ and $\pi_2$ overlap critically;
otherwise they are called orthogonal ($\pi_1 \perp \pi_2$).
Note that $\pi_1 \parallel \pi_2$ implies $\pi_1 \perp \pi_2$.
We write $P \perp Q$ if $\pi \perp \pi'$ for all $\pi \in P$ and
$\pi' \in Q$
and similarly $P \parallel Q$ if
$\pi \parallel \pi'$ for all $\pi \in P$ and $\pi' \in Q$.
If $P$ is a set of pairwise parallel redex patterns matching a term
$t$, we denote by $t \xR{P} t'$ the parallel rewrite step from $t$
to $t'$ by $P$, where $t' = t^{\pi_1\cdots\pi_n}$ if
$P = \{ \seq{\pi} \}$.
We allow $P$ to be abbreviated to a set of positions in
$t \xR{P} t'$.

We write $\langle A, \{ \to_\alpha \}_{\alpha \in I} \rangle$ to
denote the ARS $\langle A, \to \rangle$ where $\to$ is the union of
$\to_\alpha$ for all ${\alpha \in I}$. 
Let $\langle A, \{ \to_\alpha \}_{\alpha \in I} \rangle$ be an ARS and
let 
$\geqslant$ and $>$ be relations on~$I$. We write 
$\xr[\Vee\alpha_1 \cdots\,\alpha_n]{}$ for the union of $\to_\beta$
where $\beta \lablt \alpha_i$ for some $1 \leqslant i \leqslant n$.
We call $\to_\alpha$ and $\to_\beta$
\emph{decreasing}
(with respect to $\geqslant$ and $>$) if
\[
{\xl[\alpha]{}} \cdot {\xr[\beta]{}} \subseteq 
{\xr[\Vee\alpha]{*}} \cdot {\xr[\Veq\beta]{=}} \cdot
{\xr[\Vee\alpha\beta]{*}} \cdot {\xl[\Vee\alpha\beta]{*}} \cdot
{\xl[\Veq\alpha]{=}} \cdot {\xl[\Vee\beta]{*}}
\]
An ARS $\langle A, \{ \to_\alpha \}_{\alpha \in I} \rangle$ is
\emph{decreasing} if there exists a preorder
$\geqslant$ and a well-founded order $>$ such that 
$\geqslant$ and $>$ are compatible and
$\to_\alpha$ and $\to_\beta$ are decreasing for
all $\alpha, \beta \in I$ with respect to $\geqslant$ and $>$. 

The following theorem is a reformulation of a result
obtained by van Oostrom~\cite{vO94} (where $\geqslant$ is
the identity relation).
While allowing a preorder $\geqslant$ does not add power, 
it is more convenient for our purposes.

\begin{theorem}
\label{THM:dd}
Every decreasing ARS is confluent.
\qed
\end{theorem}

\section{Labelings for Rewrite Steps}
\label{LAB:one:main}

In this section we present constraints (on a labeling) such that
decreasingness of the critical peaks ensures
confluence of linear (Section~\ref{CR:l}) and left-linear
(Section~\ref{CR:ll}) TRSs.
Furthermore, we show that if two labelings satisfy these
conditions then also their lexicographic combination satisfies them.

For a local peak
\begin{equation}
\label{EQ:peak}
t = s[r_1\sigma]_p \from s[l_1\sigma]_p = s = s[l_2\sigma]_q \to
s[r_2\sigma]_q = u
\end{equation}
there are three possibilities (modulo symmetry):
\begin{itemize}
\item[\subref{FIG:peaks:p}]
$p \parallel q$
\XP,
\item[\subref{FIG:peaks:o}]
$q \leqslant p$ and $p \backslash q \in \FPos(l_2)$
\XO,
\item[\subref{FIG:peaks:v}]
$q < p$ and $p \backslash q \notin \FPos(l_2)$
\XV.
\end{itemize}
\begin{figure}
\subfloat[\label{FIG:peaks:p} \XP]{
\begin{tikzpicture}[scale=0.5]
\small
\node at (3,6) (s)            {$s$};
\node at (0,3) (t)            {$t$};
\node at (6,3) (u)            {$u$};
\node at (3,0) (v)            {$v$};
\draw[to] (s) to node[left]  {} (t);
\draw[to] (t) to node[left]  {} (v);
\draw[to] (s) to node[right] {} (u);
\draw[to] (u) to node[right] {} (v);
\end{tikzpicture}
}
\hfill
\subfloat[\label{FIG:peaks:o} \XO]{
\begin{tikzpicture}[scale=0.5]
\small
\node at (3,6) (s)            {$s$};
\node at (0,3) (t)            {$t$};
\node at (6,3) (u)            {$u$};
\node at (1,2) (t1)           {$\cdot$};
\node at (5,2) (u1)           {$\cdot$};
\node at (3,0) (v)            {};
\node at (3,1)                {\LARGE ?};
\draw[to] (s) to node[anchor=south east] {} (t);
\draw[to] (t) to node[anchor=north east] {} (t1);
\draw[to] (s) to node[anchor=south west] {} (u);
\draw[to] (u) to node[anchor=north west] {} (u1);
\end{tikzpicture}
}
\hfill
\subfloat[\label{FIG:peaks:v} \XV]{
\begin{tikzpicture}[scale=0.5]
\small
\node at (3,6) (s)            {$s$};
\node at (0,3) (t)            {$t$};
\node at (1.5,1.5) (t1)       {$t_1$};
\node at (6,3) (u)            {$u$};
\node at (3,0) (v)            {$v$};
\draw[to]   (s)  to node[anchor=south east] {} (t);
\draw[To]   (t)  to node[anchor=south east] {} (t1);
\draw[to]   (t1) to node[anchor=north east] {} (v);
\draw[to]   (s)  to node[anchor=south west] {} (u);
\draw[To]   (u)  to node[anchor=north west] {} (v);
\end{tikzpicture}
}
\caption{Three kinds of local peaks.}
\label{FIG:peaks}
\end{figure}
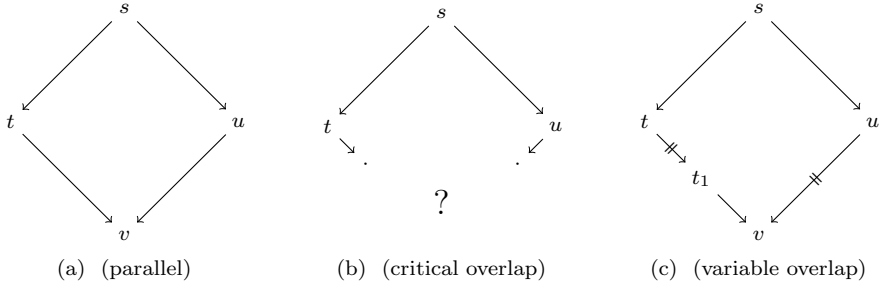
These cases are visualized in Figure~\ref{FIG:peaks}.
Figure~\ref{FIG:peaks}\subref{FIG:peaks:p}
shows the shape of a local peak where the
steps take place at parallel positions. 
Here we have $s \xr{p, l_1 \to r_1} t$ and
$u \xr{p, l_1 \to r_1} v$
as well as $s \xr{q, l_2 \to r_2} u$ and
$t \xr{q, l_2 \to r_2} v$,
i.e., the steps drawn at opposing sides in the diagram are due to
the same rules. The question mark in
Figure~\ref{FIG:peaks}\subref{FIG:peaks:o} conveys that joinability of
critical overlaps may depend on auxiliary rules.
Variable overlaps (Figure~\ref{FIG:peaks}\subref{FIG:peaks:v})
can again be joined by the rules involved in the diverging step.
More precisely, if $q'$ is the unique position in $\VPos(l_2)$ such
that $qq' \leqslant p$, $x = l_2|_{q'}$, $|l_2|_x = m$, and
$|r_2|_x = n$ then
we have $t \xr[l_1 \to r_1]{m-1}t_1$, $t_1 \xr[l_2 \to r_2]{} v$,
and $u \xr[l_1\to r_1]{n} v$.

Labelings are used to compare rewrite steps. 
In the sequel we denote the set of all rewrite steps for a TRS $\RR$
by $\STEPS{\RR}$ and elements from this set by capital Greek letters
$\Gamma$ and $\Delta$. Furthermore if
$\Gamma = s \xr{p, l \to r} t$ then
$C[\Gamma\sigma]$ denotes the rewrite step
$C[s\sigma] \xr{p'p, l \to r} C[t\sigma]$ 
for any substitution $\sigma$ and context $C$ with $C|_{p'} = \HOLE$.

\begin{definition}
\label{DEF:lab}
Let $\RR$ be a TRS. A \emph{labeling function}
$\ell\colon \STEPS{\RR} \to W$ is a mapping from rewrite steps into
some set $W$. A \emph{labeling} $(\ell,\geqslant,>)$ for $\RR$
consists of a labeling function $\ell$,
a preorder $\geqslant$, and a well-founded order $>$
such that $\geqslant$ and $>$ are compatible and for 
all rewrite steps $\Gamma, \Delta \in \STEPS{\RR}$,
contexts $C$ and substitutions $\sigma$:
\begin{enumerate}
\item
\label{DEF:lab:1}
$\ell(\Gamma) \geqslant \ell(\Delta)$ implies 
$\ell(C[\Gamma\sigma]) \geqslant \ell(C[\Delta\sigma])$,
and 
\item
\label{DEF:lab:2}
$\ell(\Gamma) > \ell(\Delta)$ implies 
$\ell(C[\Gamma\sigma]) > \ell(C[\Delta\sigma])$.
\end{enumerate}
\end{definition}

All labelings we present satisfy ${>} \subseteq {\geqslant}$, which allows
to avoid tedious case distinctions, and we assume this property henceforth.
We do so without loss of generality, because
$(({>} \cup {\geqslant})^*,>)$ 
satisfies the conditions of Definition~\ref{DEF:lab} if
$(\geqslant,>)$ does.

In the sequel $W$, $\geqslant$, and $>$ are left implicit when clear
from the context and a labeling is identified with the labeling
function $\ell$. We use the terminology that a labeling $\ell$ is 
\emph{monotone} and \emph{stable}
if properties \ref{DEF:lab:1} 
and \ref{DEF:lab:2} of Definition~\ref{DEF:lab} hold.
Abstract labels, i.e., labels that are unknown, are represented by
lowercase Greek letters $\alpha, \beta, \gamma$, and $\delta$.
We write $s \xr[\alpha]{\pi} t$ (or simply $s \to_\alpha t$)
if $\ell(s \xr{\pi} t) = \alpha$.
Often we leave the labeling $\ell$ implicit and just attach labels to
arrows. A local peak $t \from s \to u$ is called
\emph{decreasing for $\ell$} if
there are labels $\alpha$ and $\beta$ such that
$t \xl[\alpha]{} s \xr[\beta]{} u$, and $\to_{\alpha}$ and $\to_\beta$
are decreasing with respect to $\geqslant$ and $>$.
To employ Theorem~\ref{THM:dd} for TRSs, decreasingness of the ARS
$\langle \FVTERMS, \{ \to_w \}_{w \in W} \rangle$ must be shown.

\begin{figure}
\subfloat[\label{FIG:p} \XP]{
\begin{tikzpicture}[scale=0.5]
\small
\node at (3,6) (s) {$s$};
\node at (0,3) (t) {$t$};
\node at (6,3) (u) {$u$};
\node at (3,0) (v) {$v$};
\draw[to] (s) to node[above,sloped] {$\alpha$} (t);
\draw[to] (t) to node[below,sloped] {$\delta$} (v);
\draw[to] (s) to node[above,sloped] {$\beta$}  (u);
\draw[to] (u) to node[below,sloped] {$\gamma$} (v);
\end{tikzpicture}
}
\hfill
\subfloat[\label{FIG:vl} \XVL]{
\begin{tikzpicture}[scale=0.5]
\small
\node at (3,6) (s) {$s$};
\node at (0,3) (t) {$t$};
\node at (6,3) (u) {$u$};
\node at (3,0) (v) {$v$};
\draw[to] (s) to node[above,sloped] {$\alpha$} (t);
\draw[to] (t) to node[below,sloped] {$\delta$} (v);
\draw[to] (s) to node[above,sloped] {$\beta$}  (u);
\draw[to] (u) to node[below,sloped] {$\gamma$} (v);
\draw (u) to node[anchor=south,rotate=45] {$=$} (v);
\end{tikzpicture}
}
\hfill
\subfloat[\label{FIG:vll} \XVLL]{
\begin{tikzpicture}[scale=0.5]
\small
\node at (3,6) (s) {$s$};
\node at (0,3) (t) {$t$};
\node at (6,3) (u) {$u$};
\node at (3,0) (v) {$v$};
\draw[to] (s) to node[above,sloped] {$\alpha$}            (t);
\draw[to] (t) to node[below,sloped] {$\delta$}            (v);
\draw[to] (s) to node[above,sloped] {$\beta$}             (u);
\draw[To] (u) to node[below,sloped] {$\overline{\gamma}$} (v);
\end{tikzpicture}
}
\caption{Labeled local peaks.}
\label{FIG:lpeaks}
\end{figure}
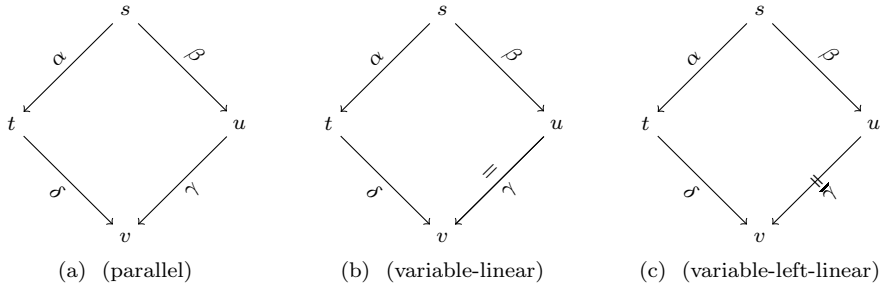

In this article we investigate conditions on a labeling such that local
peaks according to \XP and \XV are decreasing automatically. This is
desirable since in general there are infinitely many local peaks
corresponding to these cases (even if the underlying TRS has
finitely many rules).
There are also infinitely many local peaks according to \XO in general,
but for a finite TRS they are captured by the finitely many 
critical overlaps.
Still, it is undecidable if they are decreasingly
joinable~\cite{HM11}.

For later reference, Figure~\ref{FIG:lpeaks} shows labeled 
local peaks for
the case \XP (Figure~\ref{FIG:lpeaks}\subref{FIG:p}) and \XV if the
rule $l_2 \to r_2$ in local peak~\eqref{EQ:peak} is linear
(Figure~\ref{FIG:lpeaks}\subref{FIG:vl}) and
left-linear (Figure~\ref{FIG:lpeaks}\subref{FIG:vll}), respectively.
In Figure~\ref{FIG:lpeaks}\subref{FIG:vll} the expression
$\overline{\gamma}$ denotes
a sequence of labels $\seq{\gamma}$.
In the subsequent analysis we will always use the fact that the
local peaks in Figure~\ref{FIG:lpeaks} can be closed by the rules
involved in the peak (applied at opposing sides in the diagram).

\subsection{Linear TRSs}
\label{CR:l}

The next definition presents sufficient abstract conditions on a
labeling such that local peaks according to the cases \XP and \XVL in
Figure~\ref{FIG:lpeaks} are decreasing. We use the observation that
for linear TRSs the \XP case can be seen as an instance of the \XVL case
to shorten proofs.

\begin{definition}
\label{DEF:l}
Let $\ell$ be a labeling for a TRS $\RR$.
We call $\ell$ an \emph{L-labeling (for $\RR$)}
if for local peaks according to \XP and \XVL we have 
$\alpha \geqslant \gamma$ and $\beta \geqslant \delta$
in Figures~\ref{FIG:lpeaks}\subref{FIG:p}
and~\ref{FIG:lpeaks}\subref{FIG:vl}, respectively.
\end{definition}

The local diagram in Figure~\ref{FIG:combination}\subref{FIG:l1}
visualizes the conditions on an L-labeling more succinctly. 
We will use L-labelings also for left-linear TRSs,
where no conditions are required for local peaks different from
\XP and \XVL.
We call the critical peaks of a TRS $\RR$
\emph{$\Phi$-decreasing} if there exists a $\Phi$-labeling $\ell$
for $\RR$ such that the critical peaks of $\RR$ are 
decreasing for~$\ell$.
In the sequel we will introduce further labelings, e.g., 
LL-labelings and weak LL-labelings. The placeholder $\Phi$ 
avoids the need for repeating the definition of decreasingness
for these labelings.

The next theorem states that L-labelings may be used to show
confluence of linear TRSs.

\begin{theorem}
\label{THM:l}
Let $\RR$ be a linear TRS. If the critical peaks of $\RR$ are
L-decreasing then $\RR$ is confluent.
\end{theorem}
\begin{proof}
By assumption the critical peaks of $\RR$ are decreasing for some
L-labeling~$\ell$. We establish confluence of $\RR$ by
Theorem~\ref{THM:dd}, i.e., show decreasingness of the
ARS $\langle \FVTERMS, \to_\RR \rangle$ where rewrite steps are
labeled according to $\ell$. Since $\RR$ is linear, local peaks have
the shape \XP, \XVL, or \XO. By definition of an L-labeling the
former two are decreasing. Now consider a local peak
according to \XO, i.e., for the local peak
\eqref{EQ:peak} we have
$q \leqslant p$ and $p \backslash q \in \FPos(l_2)$. Let
$p' = p \backslash q$.
Then $t|_q \xl{} s|_q \xr{} u|_q$ must be an instance of
a critical peak
$l_2\mu[r_1\mu]_{p'} \xl{} l_2[l_1\mu]_{p'} = l_2\mu \xr{} r_2\mu$
which is decreasing by
assumption. By monotonicity and stability of $\ell$ we obtain 
decreasingness of the local peak \eqref{EQ:peak}.
\qed
\end{proof}

We recall the rule labeling of van Oostrom~\cite{vO08a}, 
parametrized by a mapping $i\colon \RR \to \Nat$.
Often $i$ is left implicit.
The rule labeling satisfies the constraints of an L-labeling.

\begin{lemma}
\label{LEM:lrl}
Let $\RR$ be a TRS and
$\ell^i_\RL(s \xr{\pi} t) = i(l_\pi \to r_\pi)$.
Then $(\ell^i_\RL,\geqslant_\Nat,>_\Nat)$ is an
L-labeling for $\RR$.
\end{lemma}
\begin{proof}
First we show that $(\ell^i_\RL,\geqslant_\Nat,>_\Nat)$ is a labeling.
The preorder $\geqslant_\Nat$ and the well-founded order $>_\Nat$
are compatible. Furthermore
$\ell^i_\RL(s \xr{\pi} t) = i(l_\pi \to r_\pi)$
which ensures monotonicity and stability of $\ell^i_\RL$. 
Hence $(\ell^i_\RL,\geqslant_\Nat,>_\Nat)$ is a labeling.
Next we show the properties demanded in Definition~\ref{DEF:l}.
For local peaks according to cases
\XP and \XVL we recall that the steps drawn at opposite sides in the
diagram, e.g., the steps labeled with $\alpha$ and $\gamma$ 
($\beta$ and $\delta$) in Figures~\ref{FIG:lpeaks}\subref{FIG:p}
and~\ref{FIG:lpeaks}\subref{FIG:vl}, are due to applications of the
same rule. Hence $\alpha = \gamma$ and $\beta = \delta$
in Figures~\ref{FIG:lpeaks}\subref{FIG:p}
and~\ref{FIG:lpeaks}\subref{FIG:vl}, which shows the result.
\qed
\end{proof}

Inspired by~\cite{HM11} we propose a labeling based on relative
termination.

\begin{lemma}
\label{LEM:lsn}
Let $\RR$ be a TRS and $\ell_\SN(s\to t) = s$. Then 
$\ell^\SS_\SN = (\ell_\SN, {\to_\RR^*}, {\to_\SR^+})$
is an L-labeling for $\RR$, provided
${\to_\SS} \subseteq {\to_\RR}$ and \SR is terminating.
\end{lemma}
\begin{proof}
Let ${\geqslant} = {\to_\RR^*}$ and ${>} = {\to_\SR^+}$.
First we show that $(\ell_\SN,\geqslant,>)$ is a labeling.
By definition of relative rewriting, $\geqslant$ and $>$ are
compatible and $>$ is well-founded by the termination assumption
of $\SR$. Since rewriting is closed under contexts
and substitutions, $\ell^\SS_\SN$ is monotone and stable and hence
a labeling.
Next we show the properties demanded in Definition~\ref{DEF:l}.
The assumption
${\to_\SS} \subseteq {\to_\RR}$ yields
${>} \subseteq {\geqslant}$.
Combining $\alpha = s = \beta$, $\gamma = u$, and $\delta = t$ with
$s \to_\RR t$ and $s \to_\RR u$ yields
$\alpha = \beta \geqslant \gamma, \delta$ for local peaks according
to \XP and \XVL in Figures~\ref{FIG:lpeaks}\subref{FIG:p}
and~\ref{FIG:lpeaks}\subref{FIG:vl}.
\qed
\end{proof}

The L-labeling from the previous lemma allows to establish a decrease
with respect to some steps of $\RR$.
The next lemma allows to combine L-labelings.
Let $\ell_1\colon \STEPS{\RR} \to W_1$ and
$\ell_2\colon \STEPS{\RR} \to W_2$.
Then $(\ell_1,\geqslant_1,>_1) \times (\ell_2,\geqslant_2,>_2)$ is
defined as $(\ell_1 \times \ell_2,\geqslant_{12},>_{12})$
where $\ell_1 \times \ell_2\colon \STEPS{\RR} \to W_1 \times W_2$
with
$(\ell_1 \times \ell_2)(\Gamma) = (\ell_1(\Gamma),\ell_2(\Gamma))$.
Furthermore $(x_1,x_2) \geqslant_{12} (y_1,y_2)$ if and only if
$x_1 >_1 y_1$ or $x_1 \geqslant_1 y_1$ and $x_2 \geqslant_2 y_2$
and $(x_1,x_2) >_{12} (y_1,y_2)$ if and only if
$x_1 >_1 y_1$ or $x_1 \geqslant_1 y_1$ and $x_2 >_2 y_2$.

\begin{lemma}
\label{LEM:llex}
Let $\ell_1$ and $\ell_2$ be L-labelings. Then
$\ell_1 \times \ell_2$ is an L-labeling.
\end{lemma}
\begin{proof}
First we show that $\ell_1 \times \ell_2$ is monotone and stable whenever
$\ell_1$ and $\ell_2$ are labelings. Indeed if
$(\ell_1\times\ell_2)(\Gamma) \geqslant (\ell_1\times\ell_2)(\Delta)$
then $\ell_1(\Gamma) > \ell_1(\Delta)$ or
$\ell_1(\Gamma) \geqslant \ell_1(\Delta)$ and
$\ell_2(\Gamma) \geqslant \ell_2(\Delta)$, which for all contexts~$C$ and
substitutions~$\sigma$ implies
$\ell_1(C[\Gamma\sigma]) > \ell_1(C[\Delta\sigma])$ or
$\ell_1(C[\Gamma\sigma]) \geqslant \ell_1(C[\Delta\sigma])$ and
$\ell_2(C[\Gamma\sigma]) \geqslant \ell_2(C[\Delta\sigma])$ by stability
and monotonicity of $\ell_1$ and $\ell_2$, which is equivalent to
$(\ell_1\times\ell_2)(C[\Gamma\sigma]) \geqslant
(\ell_1\times\ell_2)(C[\Delta\sigma])$.
Showing stability and monotonicity of $>$ is similar.
Since the lexicographic product satisfies
${>_{12}} \subseteq {\geqslant_{12}}$
if $\ell_1$ and $\ell_2$ are labelings we conclude
that $\ell_1 \times \ell_2$ is a labeling.
 
Next we show that $\ell_1 \times \ell_2$ satisfies the
requirements of Definition~\ref{DEF:l}.
If $\ell_1$ and $\ell_2$ are L-labelings then the diagram
of Figure~\ref{FIG:lpeaks}\subref{FIG:vl} has the shape as in
Figure~\ref{FIG:combination}\subref{FIG:l1}
and~\ref{FIG:combination}\subref{FIG:l2}, respectively. It is easy to
see that the lexicographic combination is again an L-labeling
(cf.\ Figure~\ref{FIG:combination}\subref{FIG:lex}).
\qed
\end{proof}

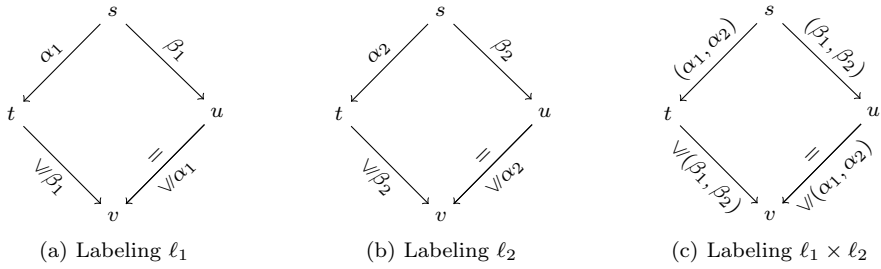
\begin{figure}
\subfloat[\label{FIG:l1}Labeling $\ell_1$]{
\begin{tikzpicture}
\small
\node (1)                      {$s$};
\node (21) [below left=of 1]   {$t$};
\node (22) [below right=of 1]  {$u$};
\node (3)  [below right=of 21] {$v$};
\draw[to] (1)  to node[above,sloped]  {$\alpha_1$} (21);
\draw[to] (1)  to node[above,sloped] {$\beta_1$}  (22);
\draw[to] (21) to node[below,sloped]  {$\Veq\beta_1$}  (3);
\draw[to] (22) to node[below,sloped] {$\Veq\alpha_1$} (3);
\draw (22) to node[anchor=south,rotate=45] {$=$} (3);
\end{tikzpicture}
}
\hfill
\subfloat[\label{FIG:l2}Labeling $\ell_2$]{
\begin{tikzpicture}
\small
\node (1)                      {$s$};
\node (21) [below left=of 1]   {$t$};
\node (22) [below right=of 1]  {$u$};
\node (3)  [below right=of 21] {$v$};
\draw[to] (1)  to node[above,sloped]  {$\alpha_2$} (21);
\draw[to] (1)  to node[above,sloped]  {$\beta_2$}  (22);
\draw[to] (21) to node[below,sloped]  {$\Veq\beta_2$}  (3);
\draw[to] (22) to node[below,sloped]  {$\Veq\alpha_2$} (3);
\draw (22) to node[anchor=south,rotate=45] {$=$}
 (3);
\end{tikzpicture}
}
\hfill
\subfloat[\label{FIG:lex}Labeling $\ell_1 \times \ell_2$]{
\begin{tikzpicture}
\small
\node (1)                      {$s$};
\node (21) [below left=of 1]   {$t$};
\node (22) [below right=of 1]  {$u$};
\node (3)  [below right=of 21] {$v$};
\draw[to] (1)  to node[above,sloped]
 {$(\alpha_1,\alpha_2)$} (21);
\draw[to] (1)  to node[above,sloped]
 {$(\beta_1,\beta_2)$}   (22);
\draw[to] (21) to node[below,sloped]
 {$\Veq(\beta_1,\beta_2)$}   (3);
\draw[to] (22) to node[below,sloped]
 {$\Veq(\alpha_1,\alpha_2)$} (3);
\draw (22) to node[anchor=south,rotate=45] {$=$} (3);
\end{tikzpicture}
}
\caption{Lexicographic combination of L-labelings.}
\label{FIG:combination}
\end{figure}

\subsection{Left-linear TRSs}
\label{CR:ll}

For left-linear TRSs the notion of an LL-labeling is introduced.
The following definition exploits that 
Figure~\ref{FIG:lpeaks}\subref{FIG:vl} is an instance of
Figure~\ref{FIG:lpeaks}\subref{FIG:vll}.

\begin{definition}
\label{DEF:ll}
\label{DEF:wll}
A labeling $\ell$ for a TRS $\RR$ is an \emph{LL-labeling (for $\RR$)} if
\begin{enumerate}
\item
\label{DEF:ll1}
in Figure~\ref{FIG:lpeaks}\subref{FIG:p},
$\alpha \geqslant \gamma$ and $\beta \geqslant \delta$,
\item
\label{DEF:ll2}
in Figure~\ref{FIG:lpeaks}\subref{FIG:vll},
$\alpha \geqslant \overline{\gamma}$
and $\beta \geqslant \delta$ for all permutations of the
rewrite steps of $u \xR{} v$,
where $\alpha \geqslant \overline{\gamma}$ means
$\alpha \geqslant \gamma_i$ for $1 \leqslant i \leqslant n$, and
\item
\label{DEF:ll3}
in Figure~\ref{FIG:lpeaks}\subref{FIG:vll},
$\alpha > \overline{\gamma}$
for some permutation of the rewrite steps of $u \xR{} v$,
where $\alpha > \overline{\gamma}$ means $\alpha \geqslant \gamma_1$
and $\alpha > \gamma_i$ for $2 \leqslant i \leqslant n$.
\end{enumerate}
A labeling $\ell$ is a \emph{weak LL-labeling} if the
first two conditions are satisfied.
\end{definition}

We strengthened the definition of (weak) LL-labelings from~\cite{ZFM11}.
All labelings proposed in~\cite{ZFM11} satisfy the stronger conditions.
Considering \emph{all} permutations in
condition~\ref{DEF:ll2} of Definition~\ref{DEF:wll}
is necessary to ensure that the lexicographic combination
of two weak LL-labelings again is a weak LL-labeling
(cf.\ Lemma~\ref{LEM:lllex}). Furthermore, this condition facilitates
their use for parallel rewriting (Section~\ref{LAB:par:main}).

\begin{remark}
\label{REM:lvswll}
The L-labelings presented so far (cf.\ Lemmata~\ref{LEM:lrl}
and~\ref{LEM:lsn}) are weak LL-labelings. 
\end{remark}

The next theorem states that LL-labelings allow to show
confluence of left-linear TRSs.

\begin{theorem}
\label{THM:ll}
Let $\RR$ be a left-linear TRS. If the critical peaks of $\RR$ are
LL-decreasing then $\RR$ is confluent.
\end{theorem}
\begin{proof}
By assumption the critical peaks of $\RR$ are decreasing for some
LL-labeling~$\ell$. We establish confluence of $\RR$ by
Theorem~\ref{THM:dd}, i.e., we show decreasingness of the ARS
$\langle \FVTERMS, \to_\RR \rangle$ by labeling rewrite steps
according to $\ell$. By definition of an LL-labeling local peaks
according to \XP and \XVLL are decreasing. The
reasoning for local peaks according to \XO is the same as in the
proof of Theorem~\ref{THM:l}.
\qed
\end{proof}

The rule labeling from Lemma~\ref{LEM:lrl} is 
a weak LL-labeling but
not an
LL-labeling since in Figure~\ref{FIG:lpeaks}\subref{FIG:vll} we have
$\alpha = \gamma_i$ for $1 \leqslant i \leqslant n$ which does not
satisfy $\alpha > \overline{\gamma}$ if $n > 1$.
(See also \cite[Example~9]{HM11}.) 
We return to this problem and propose two solutions (in 
Sections~\ref{LAB:ll:per} and~\ref{LAB:par:main})
after presenting simpler (weak) LL-labelings based on measuring
duplicating steps (Section~\ref{LAB:ll:dup}),
the context above the contracted redex (Section~\ref{LAB:ll:con}),
and the contracted redex (Section~\ref{LAB:ll:red}).

\subsubsection{Measuring duplicating steps}
\label{LAB:ll:dup}

The L-labeling from Lemma~\ref{LEM:lsn} can be adapted to an
LL-labeling.

\begin{lemma}
\label{LEM:llsn}
Let $\RR$ be a TRS. Then $\ell^\RRd_\SN$ is an LL-labeling,
provided $\RRdnd$ is terminating.
\end{lemma}
\begin{proof}
By Theorem~\ref{THM:rt} the relative TRS $\RRdnd$ is terminating if
and only if $\REL{\RRd}{\RR}$ is terminating. Hence
$(\ell^\RRd_\SN,\geqslant,>)$ is a labeling by Lemma~\ref{LEM:lsn}.
Here ${\geqslant} = {\to^*_\RR}$ and ${>} = {\to^+_\REL{\RRd}{\RR}}$.
Since $\ell_\SN(s \to t) = s$, we have $\alpha = \beta$ in
Figures~\ref{FIG:lpeaks}\subref{FIG:p}
and~\ref{FIG:lpeaks}\subref{FIG:vll}.
We have ${>} \subseteq {\geqslant}$. Hence
$\alpha \geqslant \gamma$ and $\alpha \geqslant \delta$ in
Figure~\ref{FIG:lpeaks}\subref{FIG:p} and, 
if $l_2 \to r_2$ in 
local peak \eqref{EQ:peak} is linear,
also in Figure~\ref{FIG:lpeaks}\subref{FIG:vll} as 
$\overline{\gamma}$ is empty or
$\overline{\gamma} = \gamma$ in this case.
If $l_2 \to r_2$ is not linear then it must be duplicating and hence 
$\alpha > \gamma_i$ for $1 \leqslant i \leqslant n$.
Because
$\alpha \geqslant \delta$,
$\ell^\RRd_\SN$ is an LL-labeling for~$\RR$.
\qed
\end{proof}

To combine the previous lemma with the rule labeling we study how
different labelings can be combined.

\begin{lemma}
\label{LEM:lllex}
Let $\ell_1$ be an LL-labeling and let $\ell_2$ be a weak LL-labeling.
Then $\ell_1 \times \ell_2$ and $\ell_2 \times \ell_1$ are LL-labelings.
\end{lemma}
\begin{proof}
By the proof of Lemma~\ref{LEM:llex} $\ell_1 \times \ell_2$ and
$\ell_2 \times \ell_1$ are labelings.
The only interesting case of \XVLL is when
$l_2 \to r_2$ in local peak~\eqref{EQ:peak} is non-linear, i.e.,
$\overline{\gamma}$ contains more than one element.
First we show that $\ell_1 \times \ell_2$ is an LL-labeling. 
Here labels according to $\ell_1$ are suffixed with the subscript $1$
and similarly for $\ell_2$.
Recall Figure~\ref{FIG:lpeaks}\subref{FIG:vll}.
Let us first deal with Definition~\ref{DEF:ll}(\ref{DEF:ll2}).
We have $\alpha_1 \geqslant \overline{\gamma}_1$,
$\beta_1 \geqslant \delta_1$,
$\alpha_2 \geqslant \overline{\gamma}_2$ and $\beta_2 \geqslant \delta_2$,
which yields
$(\beta_1,\beta_2) \geqslant (\delta_1,\delta_2)$,
$(\alpha_1,\alpha_2) \geqslant (\gamma_{1i},\gamma_{2i})$ for all
$1\leqslant i\leqslant n$,
by the definition of the lexicographic product.
Next we consider Definition~\ref{DEF:ll}(\ref{DEF:ll3}).
By assumption we have  $\alpha_1 > \overline{\gamma}_1$, and 
$\alpha_2 \geqslant \overline{\gamma}_2$,
which yields the desired
$(\alpha_1,\alpha_2) \geqslant (\gamma_{11},\gamma_{21})$,
$(\alpha_1,\alpha_2) > (\gamma_{1i},\gamma_{2i})$ 
for $2 \leqslant i \leqslant n$.
In the proof for $\ell_2 \times \ell_1$ the assumptions yield
$(\beta_2,\beta_1) \geqslant (\delta_2,\delta_1)$ and
$(\alpha_2,\alpha_1) \geqslant (\gamma_{2i},\gamma_{1i})$ for
$1\leqslant i\leqslant n$
for Definition~\ref{DEF:ll}(\ref{DEF:ll2}) and additionally
$(\alpha_2,\alpha_1) > (\gamma_{2i},\gamma_{1i})$ for
$2 \leqslant i \leqslant n$ for Definition~\ref{DEF:ll}(\ref{DEF:ll3}).
\qed
\end{proof}

\begin{remark}
\label{REM:wlllex}
If $\ell_1$ and $\ell_2$ are weak LL-labelings then so are
$\ell_1 \times \ell_2$ and $\ell_2 \times \ell_1$.
Furthermore, LL-labelings are also weak LL-labelings by definition.
In particular LL-labelings can be composed lexicographically.
\end{remark}

From Theorem~\ref{THM:ll} and Lemmata~\ref{LEM:llsn}
and~\ref{LEM:lllex} we obtain the following result.

\begin{corollary}
\label{COR:rtd}
Let $\RR$ be a left-linear TRS. If $\RRdnd$ is terminating and all
critical peaks of $\RR$ are weakly LL-decreasing then $\RR$ is
confluent.
\end{corollary}
\begin{proof}
By Lemma~\ref{LEM:llsn} $\ell^\RRd_\SN$ is an LL-labeling.
By assumption the critical peaks of $\RR$ are decreasing for some
weak LL-labeling $\ell$. By Lemma~\ref{LEM:lllex}
also $\ell^\RRd_\SN \times \ell$ is an LL-labeling. 
It remains to show decreasingness of the critical peaks for 
$\ell^\RRd_\SN \times \ell$. This is obvious since for terms $s$, $t$,
$u$ with $s \to_\RR t \to_\RR u$ we have
$\ell^\RRd_\SN(s \to t) \geqslant \ell^\RRd_\SN(t \to u)$.
Hence decreasingness for $\ell$ implies decreasingness for
$\ell^\RRd_\SN \times \ell$.
Confluence of $\RR$ follows from Theorem~\ref{THM:ll}.
\qed
\end{proof}

We revisit the example from the introduction.

\begin{example}
\label{EX:OO03d}
Recall the TRS $\RR$ from Example~\ref{EX:OO03}. The polynomial
interpretation 
\begin{xalignat*}{4}
+_\Nat(x,y) &= x + y &
\m{s}_\Nat(x) &= x + 1 & 
\times_\Nat(x,y) &= x^2 + xy + y^2 & 
\m{sq}_\Nat(x) &= 3x^2 + 1
\end{xalignat*}
shows termination of $\RRdnd$.
It is easy to check that $\ell^i_\RL$ with
$i(3) = i(6) = 2$, $i(4) = i(10) = 1$, and $i(l \to r) = 0$
for all other rules $l \to r \in \RR$ establishes decreasingness of
the 34 critical peaks. 
We consider two selected critical peaks (where the applied 
rewrite rule is indicated above the arrow in parentheses). The peaks
\newcommand{\xrr}[2]{\xr[#1]{\makebox[3mm]{$\scriptstyle #2$}}}
\newcommand{\xll}[2]{\xl[#1]{\makebox[3mm]{$\scriptstyle #2$}}}
\begin{align*}
t_1 &= x + ((y + z) + w) \xll{0}{(1)} x + (y + (z + w)) \xrr{0}{(1)}
(x + y) + (z + w) = u_1 \\
t_2 &= \m{s}(x) \times \m{s}(x) \xll{2}{(3)} \m{sq}(\m{s}(x)) \xrr{1}{(4)}
(x \times x) + \m{s}(x + x) = u_2
\end{align*}
can be joined decreasingly as follows:
\begin{align*}
t_1 &\xrr{0}{(2)} x + (y + (z + w)) \xll{0}{(2)} u_1 \\
t_2 &\xrr{1}{(10)} (x \times \m{s}(x)) + \m{s}(x)
\xrr{0}{(9)} (x + (x \times x)) + \m{s}(x)
\xrr{0}{(2)} x + ((x \times x) + \m{s}(x)) \\
&\xrr{0}{(8)} x + (\m{s}(x \times x) + x)
\xll{0}{(2)} (x + \m{s}(x \times x)) + x
\xll{0}{(5)} (\m{s}(x \times x) + x) + x \\
&\xll{0}{(1)} \m{s}(x \times x) + (x + x)
\xll{0}{(8)} u_2
\end{align*}
\end{example}

The next example is concise and constitutes a minimal example
to familiarize the reader with Corollary~\ref{COR:rtd}.

\begin{example}
\label{EX:mot}
Consider the TRS $\RR$ consisting of the three rules
\begin{xalignat*}{3}
1\colon~\m{b} &\to \m{a} &
2\colon~\m{a} &\to \m{b} &
3\colon~\m{f}(\m{g}(x,\m{a})) &\to \m{g}(\m{f}(x),\m{f}(x))
\end{xalignat*}
We have $\RRd = \{ 3 \}$ and $\RRnd = \{ 1, 2 \}$.
Termination of $\RRdnd$ can be established by LPO with precedence
$\m{a} \sim \m{b}$ and $\m{f} > \m{g}$. 
The rule labeling that takes the rule numbers as labels shows the
only critical peak decreasing, i.e.,
$\m{f}(\m{g}(x,\m{b})) \xl[2]{} \m{f}(\m{g}(x,\m{a})) \to_3
\m{g}(\m{f}(x),\m{f}(x))$ and
$\m{f}(\m{g}(x,\m{b})) \to_1 \m{f}(\m{g}(x,\m{a})) \to_3
\m{g}(\m{f}(x),\m{f}(x))$.
Hence we obtain the confluence of $\RR$ by Corollary~\ref{COR:rtd}.
\end{example}

\begin{remark}
Using $\ell^i_\RL(\cdot) = 0$ as weak LL-labeling, 
Corollary~\ref{COR:rtd} gives a condition (termination of $\RRdnd$)
such that $t \to^= u$ or $u \to^= t$ for all critical pairs $t \cp u$ 
implies confluence of a left-linear TRS $\RR$. This partially answers
one question in the RTA list of open problems~\#13.%
\footnote{\url{%
http://www.cs.tau.ac.il/~nachum/rtaloop/problems/13.html%
}}
\end{remark}

\subsubsection{Measuring the context above the contracted redex}
\label{LAB:ll:con}

In~\cite[Example~20]{vO08a} van Oostrom suggests to count function
symbols above the contracted redex, demands that this measurement
decreases for variables that are duplicated, and combines
this with the rule labeling. Consequently local peaks according to
Figure~\ref{FIG:lpeaks}\subref{FIG:vll} are decreasing. Below we
exploit this idea but incorporate the following beneficial
generalizations. First, we do not restrict to counting function
symbols (which has been adopted and extended by Aoto in~\cite{A10})
but represent the constraints as a relative termination problem.
This abstract formulation allows to strictly subsume the criteria
from~\cite{vO08a,A10} (see Section~\ref{ASS:main})
because more advanced techniques than counting symbols can be applied
for proving termination. Additionally, our setting also allows to
weaken these constraints significantly
(cf.\ Lemma~\ref{LEM:sstar}).

The next example motivates the need for an LL-labeling
that does not require termination of $\RRdnd$. 

\begin{example}
\label{EX:mot2}
Consider the TRS $\RR$ consisting of the six rules
\begin{xalignat*}{3}
\m{f}(\m{h}(x)) &\to \m{h}(\m{g}(\m{f}(x),x,\m{f}(\m{h}(\m{a})))) &
\m{f}(x) &\to \m{a} &
\m{a} &\to \m{b} \\
\m{h}(x) &\to \m{c} &
\m{b} &\to \bot &
\m{c} &\to \bot
\end{xalignat*}
Since the duplicating rule admits an infinite sequence,
Corollary~\ref{COR:rtd} cannot succeed.
\end{example}

In the sequel we let $\GG$ be the signature consisting of unary
function symbols $\seq{f}$ for every $n$-ary function symbol
$f \in \FF$. 

\begin{definition}
\label{DEF:con}
Let $x \in \VV$. We define a partial mapping $\star$ from
terms in the original signature and positions
$\FVTERMS \times \NAT_+^*$ to terms in $\TERMS{\GG}{\VV}$
as follows:
\[
\ST{f(\seq{t}),p} = \begin{cases}
f_i(\ST{t_i,q}) & \text{if $p = iq$} \\
x & \text{if $p = \epsilon$}
\end{cases}
\]
For a TRS $\RR$ we abbreviate
$\REL{\STG{\RR}}{\STE{\RR}}$ by $\ST{\RR}$.
Here, for ${\gtrsim} \in \{ {>}, {=} \}$,
$\STC{\RR}$ consists of all rules $\ST{l,p} \to \ST{r,q}$
such that $l \to r \in \RR$, $l|_p = r|_q = y \in \VV$,
and $|r|_y \gtrsim 1$.
\end{definition}

The next example illustrates the transformation \ST{$\cdot$}.

\begin{example}
\label{EX:star}
Consider the TRS $\RR$ from Example~\ref{EX:mot2}.
The relative TRS $\ST{\RR} = \REL{\STG{\RR}}{\STE{\RR}}$ 
consists of the TRS $\STG{\RR}$ with rules
\begin{xalignat*}{2}
\m{f_1}(\m{h_1}(x)) &\to \m{h_1}(\m{g_1}(\m{f_1}(x))) &
\m{f_1}(\m{h_1}(x)) &\to \m{h_1}(\m{g_2}(x))
\end{xalignat*}
and the TRS $\STE{\RR}$ which is empty. 
\end{example}

Due to the next lemma a termination proof of $\ST{\RR}$ yields an
LL-labeling.

\begin{lemma}
\label{LEM:star}
Let $\RR$ be a TRS and $\ell_\star(s \xr{\pi} t) = \star(s,p_\pi)$.
Then $(\ell_\star,\geqslant,>)$ is an LL-labeling,
provided $(\geqslant,>)$ is a monotone reduction pair,
${\STG{\RR}} \subseteq {>}$, and
${\STG{\RR} \cup \STE{\RR}} \subseteq {\geqslant}$.
\end{lemma}
\begin{proof}
Because $(\geqslant,>)$ is a monotone reduction pair,
$(\ell_\star,\geqslant,>)$ is a labeling for~$\RR$.
Note that monotonicity and stability are with respect to the
signature~$\GG$.
To see that the constraints of Definition~\ref{DEF:ll} are satisfied
we argue as follows. For Figure~\ref{FIG:lpeaks}\subref{FIG:p} we have
$\alpha = \gamma$ and $\beta = \delta$ because the steps drawn at
opposing sides in the diagram take place at the same positions and
the function symbols above these positions stay the same. 
Next we consider Figure~\ref{FIG:lpeaks}\subref{FIG:vl}, i.e., the
right-linear case. Recall the local peak
\eqref{EQ:peak}. 
Again we have $\beta = \delta$ because $q < p$.
To see $\alpha \geqslant \gamma$ 
consider the step $s \xr{q,l_2\to r_2} u$ and
let $q'$ be the unique position in $\VPos(l_2)$ such that
$qq'r = p$ with $x = l_2|_{q'}$ for some position $r$.
If $|r_2|_x = 0$ then there is
no step and we are done. Otherwise let $q''$ be the position in
$r_2$ with $|r_2|_{q''} = x$.
By construction $\STE{\RR}$ contains the rule
$\star(l_2,q') \to \star(r_2,q'')$. 
Combining the assumption $\STE{\RR} \subseteq {\geqslant}$ with
monotonicity and stability of $\ell_\star$ yields
$\star(s,p) \geqslant \star(u,qq''r)$, i.e., $\alpha \geqslant \gamma$.
Next we consider Figure~\ref{FIG:lpeaks}\subref{FIG:vll} for
the duplicating case. Recall the local peak \eqref{EQ:peak}. 
Again we have $\beta = \delta$ because $q < p$.
To see $\alpha > \overline{\gamma}$
(for any permutation of the steps)
consider the step $s \xr{q,l_2\to r_2} u$ and
let $q'$ be the unique position in $\VPos(l_2)$ such that $qq'r = p$
for some position $r$. Let
$x = l_2|_{q'}$ and $Q = \{ \seq{q'} \}$ with $r_2|_{q_i'} = x$.
Then $P = \{ qq_i'r \mid q_i' \in Q \}$
is the set of descendants of $p$.
By construction $\STG{\RR}$ contains all rules
$\star(l_2,q') \to \star(r_2,q_i')$ for $1 \leqslant i \leqslant n$.
Combining the assumption $\STG{\RR} \subseteq {>}$ with monotonicity
and stability of $\ell_\star$ yields
$\star(s,p) > \star(u,p_i')$ for $p_i' \in P$.
Since $u \xR{P}{} v$ we obtain $\alpha > \gamma_i$ for
$1 \leqslant i \leqslant n$  and hence
the desired $\alpha > \overline{\gamma}$.
\qed
\end{proof}

\begin{remark}
It is also possible to formulate Lemma~\ref{LEM:star} as a relative
termination criterion without the use of a monotone reduction pair.
However, the monotone reduction pair may admit more labels to be
comparable (in the critical diagrams) because of the inclusions
${\STG{\RR}} \subseteq {>}$ and
${\STG{\RR} \cup \STE{\RR}} \subseteq {\geqslant}$.
\end{remark}

From Lemma~\ref{LEM:star} we obtain the following corollary.

\begin{corollary}
\label{COR:rt}
Let $\RR$ be a left-linear TRS and let $\ell$ be a weak LL-labeling.
Let $\ell_\star\ell$ denote $\ell \times \ell_\star$ or
$\ell_\star \times \ell$.
Let $({\geqslant},{>})$ be a monotone reduction pair
showing termination of $\ST{\RR}$.
If the
critical peaks of $\RR$ are decreasing for $\ell_\star\ell$
then $\RR$ is confluent.
\end{corollary}
\begin{proof}
The function $\ell_\star$ is an LL-labeling by Lemma~\ref{LEM:star}.
Lemma~\ref{LEM:lllex} yields that $\ell_\star\ell$ is an 
LL-labeling. By assumption the critical peaks are decreasing for
$\ell_\star\ell$ and hence Theorem~\ref{THM:ll} yields the
confluence of $\RR$.
\qed
\end{proof}

The next example illustrates the use of Corollary~\ref{COR:rt}.

\begin{example}
\label{EX:mot2c}
We show confluence of the TRS $\RR$ from Example~\ref{EX:mot2}.
Termination of $\ST{\RR}$ (cf.\ Example~\ref{EX:star}) is
easily shown, e.g., the polynomial interpretation
\begin{xalignat*}{3}
\m{f_1}_\Nat(x) &= 2x &
\m{g_1}_\Nat(x) &= \m{g_2}_\Nat(x) = x &
\m{h_1}_\Nat(x) &= x + 1
\end{xalignat*}
orients both rules in $\STG{\RR}$ strictly.
To show decreasingness of the three critical peaks
(two of which are symmetric) we use the labeling
$\ell_\star \times \ell^i_\RL$ with 
${i(\m{f}(\m{h}(x)) \to
\m{h}(\m{g}(\m{f}(x),x,\m{f}(\m{h}(\m{a}))))) = 1}$
and all other rules receive label $0$.
For the moment we label a step $s \xr{\pi} t$ with 
the interpretation of $\star(s,p_\pi)$.
E.g., a step $\m{f}(\m{h}(\m{b})) \to \m{f}(\m{h}(\bot))$ is
labeled $2x + 2$ since
$\star(\m{f}(\m{h}(\m{b})),11) = \m{f_1}(\m{h_1}(x))$ and
$[\m{f_1}(\m{h_1}(x))]_\Nat = 2x + 2$.
The critical peak
${\m{h}(\m{g}(\m{f}(x),x,\m{f}(\m{h}(\m{a})))) \xl[x,1]{}
\m{f}(\m{h}(x)) \xr[{x,0}]{} \m{a}}$ is closed decreasingly by 
\[
\m{h}(\m{g}(\m{f}(x),x,\m{f}(\m{h}(\m{a}))))
\xr[{x,0}]{} \m{c} 
\xr[{x,0}]{}
\bot \xl[x,0]{} \m{b}
\xl[x,0]{} \m{a} 
\]
and the critical peak $\m{h}(\m{g}(\m{f}(x),x,\m{f}(\m{h}(\m{a})))) 
\xl[x,1]{} \m{f}(\m{h}(x)) \xr[2x,0]{} \m{f}(\m{c})$ is closed
decreasingly by
\[
\m{h}(\m{g}(\m{f}(x),x,\m{f}(\m{h}(\m{a})))) 
\xr[x,0]{} \m{c} 
\xr[x,0]{} \bot
\xl[x,0]{} \m{b} 
\xl[x,0]{} \m{a} 
\xl[x,0]{} \m{f}(\m{c})
\]
which allows to prove confluence of $\RR$ by Corollary~\ref{COR:rt}. 
\end{example}

By definition of $\alpha > \overline{\gamma}$
(cf.\ Definition~\ref{DEF:ll})
we observe that the definition of $\ST{\RR}$ can be relaxed. 
If $l_2 \to r_2$ with ${l_2}|_{q'} = x \in \VV$ and
$\{ \seq{q'} \}$ are the positions of the variable $x$ in $r_2$ then
it suffices if $n-1$ instances of $\ST{l_2,q'} \to \ST{r_2,q'_i}$ are
put in $\STG{\RR}$ while one $\ST{l_2,q'} \to \ST{r_2,q'_j}$ can be
put in $\STE{\RR}$ (since the steps labeled $\overline{\gamma}$ in
Figure~\ref{FIG:lpeaks}\subref{FIG:vll} are at parallel
positions we can choose the first closing step such that
$\alpha \geqslant \gamma_1$). This improved version of $\ST{\RR}$ is
denoted by $\STT{\RR} = \REL{\STTG{\RR}}{\STTE{\RR}}$.
We obtain the following variant of Lemma~\ref{LEM:star}.

\begin{lemma}
\label{LEM:sstar}
Let $\RR$ be a TRS. Then $(\ell_\star,\geqslant,>)$
is an LL-labeling, provided
$(\geqslant,>)$ is a monotone reduction pair,
${\STTG{\RR}} \subseteq {>}$, and
${\STTG{\RR} \cup \STTE{\RR}} \subseteq {\geqslant}$.
\qed
\end{lemma}

Obviously any $\STT{\RR}$ inherits termination from $\ST{\RR}$.
The next example shows that the reverse statement does not hold. In
Section~\ref{IMP:main} we show how the intrinsic indeterminism
of $\STT{\RR}$ is eliminated in the implementation.

\begin{example}
\label{EX:OO03sstar}
Consider the TRS $\RR$ from Example~\ref{EX:OO03}.
The TRS $\STG{\RR}$ consists of the rules
\begin{xalignat*}{3}
\m{sq_1}(x) &\to \m{\times_1}(x)
& \m{sq_1}(\m{s_1}(x)) &\to \m{+_1}(\m{\times_1}(x))
& \m{\times}_1(x) &\to \m{+_1}(x)
\\
\m{sq_1}(x) &\to \m{\times_2}(x)
& \m{sq_1}(\m{s_1}(x)) &\to \m{+_1}(\m{\times_2}(x))
& \dagger\colon \m{\times}_1(x) &\to \m{+_2}(\m{\times_1}(x))
\\
&
& \m{sq_1}(\m{s_1}(x)) &\to \m{+_2}(\m{s_1}(\m{+_1}(x)))
& \dagger\colon \m{\times_2}(y) &\to \m{+_1}(\m{\times_2}(y))
\\
&
& \m{sq_1}(\m{s_1}(x)) &\to \m{+_2}(\m{s_1}(\m{+_2}(x)))
& \m{\times_2}(y) &\to \m{+_2}(y)
\end{xalignat*}
while $\STE{\RR}$ consists of the rules
\begin{xalignat*}{3}
\m{+_1}(x) &\to \m{+_1}(\m{+_1}(x)) 
& \m{+_1}(x) &\to \m{+_2}(x)
& \m{+_1}(x) &\to \m{+_1}(\m{s_1}(x))
\\
\m{+_2}(\m{+_1}(y)) &\to \m{+_1}(\m{+_2}(y))
& \m{+_2}(y) &\to \m{+_1}(y)
& \m{+_2}(\m{s_1}(y)) &\to \m{+_2}(y)
\\
\m{+_2}(\m{+_2}(z)) &\to \m{+_2}(z)
& \m{\times_1}(x) &\to \m{\times_2}(x)
& \m{\times_2}(\m{s_1}(y)) &\to \m{+_2}(\m{\times_2}(y))
\\
\m{+_1}(\m{+_1}(x)) &\to \m{+_1}(x)
& \m{\times_2}(y) &\to \m{\times_1}(y)
& \m{\times_1}(\m{s_1}(x)) &\to \m{+_1}(\m{\times_1}(x))
\\
\m{+_1}(\m{+_2}(y)) &\to \m{+_2}(\m{+_1}(y))
& \m{+_1}(\m{s_1}(x)) &\to \m{+_1}(x)
\\
\m{+_2}(z) &\to \m{+_2}(\m{+_2}(z))
& \m{+_2}(y) &\to \m{+_2}(\m{s_1}(y))
\end{xalignat*}
Let $\STD{\RR}$ denote the rules in $\STG{\RR}$ marked with $\dagger$.
Termination of $\ST{\RR}$ cannot be established
(because $\STD{\RR}$ is non-terminating) but we stress that
moving these rules into $\STE{\RR}$ yields a valid $\STT{\RR}$
which can be proved terminating by the polynomial interpretation with
\begin{xalignat*}{2}
\m{sq_1}_\Nat(x) &= x + 2 &
\m{\times_1}_\Nat(x) &= \m{\times_2}_\Nat(x) = x + 1
\end{xalignat*}
that interprets the remaining function symbols by the identity
function. We remark that Corollary~\ref{COR:rt} with the labeling
from Lemma~\ref{LEM:sstar} establishes confluence of $\RR$.
Since all reductions in the 34 joining sequences have only $+$ above
the redex and $\m{+_1}_\Nat(x) = \m{+_2}_\Nat(x) = x$,
the $\ell_\star$ labeling attaches $x$ to any of these steps. The
rule labeling that assigns $i(3) = i(6) = 2$, $i(4) = i(10) = 1$,
and $0$ to all other rules shows the 34 critical peaks decreasing.
\end{example}

\subsubsection{Measuring the contracted redex}
\label{LAB:ll:red}

Instead of the labeling $\ell_\star$, which is based on the
context above the contracted redex, one can also use the
contracted redex itself for labeling.

\begin{lemma}
\label{LEM:red}
Let $\RR$ be a TRS and $\ell_\red(s \xr{\pi} t) = s|_{p_\pi}$.
Then $(\ell_\red,\geqslant,>)$ is a weak
LL-labeling, provided $(\geqslant,>)$ is a monotone reduction pair
with $\RR \subseteq {\geqslant}$.
\end{lemma}
\begin{proof}
Because $(\geqslant,>)$ is a monotone reduction pair,
$(\ell_\red,\geqslant,>)$ is a labeling for $\RR$.
To see that the constraints of Definition~\ref{DEF:wll} are satisfied
we argue as follows.
For Figure~\ref{FIG:lpeaks}\subref{FIG:p} we have $\alpha = \gamma$ and
$\beta = \delta$. For Figure~\ref{FIG:lpeaks}\subref{FIG:vll} we have
$\alpha = \gamma_1 = \cdots = \gamma_n$ (since
the same redex is contracted) and $\beta \geqslant \delta$ by the
assumption $\RR \subseteq {\geqslant}$ and monotonicity and stability
of $\geqslant$.
\qed
\end{proof}

The following definition collects the constraints, such that variable
overlaps can be made decreasing.

\begin{definition}
For a TRS $\RR$ let $\RR^\red = \{ l \to x \mid
\text{$l \to r \in \RR$ and $|r|_x > 1$} \}$.
\end{definition}

Due to the next result a termination proof of $\RR^\red/\RR$
enables a weak LL-labeling to establish confluence.

\begin{corollary}
\label{COR:red}
Let $\RR$ be a left-linear TRS and let $\ell$ be a weak LL-labeling.
Let $(\geqslant,>)$ be a simple monotone reduction pair showing
termination of $\RR^\red/\RR$. If the critical peaks of $\RR$ are
decreasing for $\ell_\red \times \ell$ then $\RR$ is confluent.
\end{corollary}

\begin{proof}
Note that $\ell_\red\times\ell$ is a weak LL-labeling
(cf.\ Remark~\ref{REM:wlllex}),
which shows the peaks in Figure~\ref{FIG:lpeaks}\subref{FIG:p} and
Figure~\ref{FIG:lpeaks}\subref{FIG:vl} decreasing.
For the duplicating case of Figure~\ref{FIG:lpeaks}\subref{FIG:vll} we 
inspect the labels with regard to $\ell_\red$.
Consider the local peak~\eqref{EQ:peak}.
Clearly, $\beta = l_2\sigma$ and $\alpha = l_1\sigma$.
Since $\gamma_i = \alpha$, we want to establish $\beta > \alpha$.
To this end let $q' \in \VPos(l_2)$
such that $qq'r = p$ and $x = l_2|_{q'}$.
Note that $l_2 \to x \in \RR^\red$ because we are in the
duplicating case. Hence the
relative termination assumption
gives $l_2 > x$, and $l_2\sigma > x\sigma$ is obtained by stability. 
Now as $x\sigma|_{r} = l_1\sigma$ the desired $\beta > \alpha$ follows
from simplicity of the reduction pair since
$l_2\sigma > x\sigma \geqslant l_1\sigma$.
Combining $\ell_\red$ lexicographically with a weak LL-labeling $\ell$
into $\ell_\red\times\ell$ maintains decreasingness.
\qed
\end{proof}

\begin{remark}
Note that the labeling $\ell_\red\times\ell$ from Corollary~\ref{COR:red}
is not an LL-labeling. The point is that there are multiple ways of
ensuring decreasingness of Figure~\ref{FIG:lpeaks}\subref{FIG:vll}.
For LL-labelings, we use $\alpha > \overline\gamma$, while in
Corollary~\ref{COR:red}, $\beta > \gamma_i$ for $1 \leqslant i \leqslant n$
does the job. This is also the reason why $\ell \times \ell_\red$
cannot be used in Corollary~\ref{COR:red}. Consider the
TRS with the rules $1:\m{f}(x) \to \m{g}(x,x)$ and 
$2:\m{a }\to \m{b}$. Let $\ell_\RL$
be the rule labeling attaching the rule numbers as labels.
Then the variable overlap is not decreasing for
$\ell_\RL \times \ell_\red$.
\end{remark}

We demonstrate Corollary~\ref{COR:red} on the TRS from
Example~\ref{EX:mot}.

\begin{example}
Consider the TRS from Example~\ref{EX:mot}. The polynomial interpretation
\begin{xalignat*}{3}
\m{g}_\Nat(x,y) &= 2x + 2y + 1 &
\m{a}_\Nat = \m{b}_\Nat &= 0 &
\m{f}_\Nat(x) &= x^2
\end{xalignat*}
establishes relative termination of
$\{ \m{f}(\m{g}(x,\m{a})) \to x \}/\RR$ and
shows the critical peak decreasing when
labeling steps with the 
pair obtained by the interpretation of the redex 
and the rule labeling, i.e.,
$t = \m{f}(\m{g}(x,\m{b})) \xl[0,2]{} \m{f}(\m{g}(x,\m{a}))
\xr[(2x+1)^2,3]{} \m{g}(\m{f}(x),\m{f}(x)) = u$ 
for the peak and
$t \xr[0,1]{} \m{f}(\m{g}(x,\m{a})) \xr[(2x+1)^2,3]{} u$
for the join.
\end{example}

\subsubsection{Exploiting Persistence}
\label{LAB:ll:per}

In this section we show how to exploit persistence of
confluence~\cite{AT97,FZM11} to enhance the applicability of L-labelings
to certain duplicating left-linear TRSs. Compared to 
Sections~\ref{LAB:ll:dup}--\ref{LAB:ll:red},
where variable overlaps were closed decreasingly
by a relative termination criterion, here persistence
arguments are employed to avoid reasoning about variable overlaps at
duplicating variable positions at all.
To this end we recall order-sorted TRSs.

\begin{definition}
Let $S$ be a set of sorts equipped with a partial order $\leq$.
A signature $\FF$ and a set of variables~$\VV$ are
$S$-sorted if
every $n$-ary function symbol $f \in \x{F}$ is equipped with a
sort declaration $\alpha_1 \times \dots \times \alpha_n \to \alpha$ where 
$\seq{\alpha}, \alpha \in S$ and every variable $x \in \x{V}$ has exactly
one sort $\alpha \in S$.
We write $S(f) = \alpha$,
$S(f,i) = \alpha_i$ for $1 \leqslant i \leqslant n$, and
$S(x) = \alpha$, respectively.
We let $\x{V}_\alpha = \{ x \in \x{V} \mid S(x) = \alpha \}$ and
require that $\x{V}_\alpha$ is infinite for all $\alpha \in S$.
The set of $S$-sorted terms, $\x{T}_S(\x{F},\x{V})$,
is the union of the sets $\x{T}_\alpha(\x{F},\x{V})$
for $\alpha \in S$ that are inductively
defined as follows: $\x{V}_\alpha \subseteq \x{T}_\alpha(\x{F},\x{V})$ and
$f(\seq{t}) \in \x{T}_\alpha(\x{F},\x{V})$ whenever 
$f \in \x{F}$ has sort declaration
$\alpha_1 \times \dots \times \alpha_n \to \alpha$ and
$t_i \in \x{T}_{{\leq}\alpha_i}(\x{F},\x{V})$ for all
$1 \leqslant i \leqslant n$.
Here $\x{T}_{{\leq}\alpha}(\x{F},\x{V})$ is the union of all
$\x{T}_\beta(\x{F},\x{V})$ for $\beta \leq \alpha$.
\end{definition}

The notion of $S$-sorted terms properly extends many-sorted terms.
Indeed, if we let $\leq$ be the identity relation then
$\x{T}_{{\leq}\alpha}(\x{F},\x{V}) = \x{T}_{\alpha}(\x{F},\x{V})$,
which means that the $i$-th argument of $f$ in an $S$-sorted
term must have sort $S(f,i)$.

\begin{definition}
We extend $S(\cdot)$ and $S(\cdot,\cdot)$ to
$S$-sorted terms $t$ and non-root positions of $t$.
If $t = f(\seq{t})$ then $S(t) = S(f)$,
$S(t,i) = S(f,i)$, and $S(t,ip) = S(t_i,p)$ for $p \neq \epsilon$.
If $t = x \in \x{V}$ then $S(t) = S(x)$.
\end{definition}

\begin{example}
Let $S = \{ 0, 1, 2 \}$ with $0 \leq 1$ and consider the
sort declarations $\m{f} : 1 \to 2$ and $x : 0$. Then
$t = \m{f}(x) \in \x{T}_S(\{ \m{f} \},\{ x \})$,
$S(t) = 2$, $S(t,1) = 1$, and 
$S(t|_1) = 0 \leq 1$.
\end{example}

One easily observes that $S(t,p)$ defines
the maximal sort induced by the context $t[\square]_p$:
a term $t[u]_p$ is $S$-sorted if and only if
$u \in \x{T}_{{\leq}S(t,p)}(\x{F},\x{V})$. Consequently,
we have $S(t|_p) \leq S(t,p)$
for all non-root positions $p$ of $t$.

We are particularly interested in the case where rewriting restricted
to $S$-sorted terms coincides with ordinary rewriting with initial
terms restricted to $S$-sorted ones. This property is captured
by $S$-compatible TRSs.

\begin{definition}
A TRS $\RR$ is \emph{$S$-compatible} if for every rule
$l \to r \in \x{R}$ there exists a sort $\alpha \in S$ such that
$l \in \x{T}_\alpha(\x{F},\x{V})$ and
$r \in \x{T}_{{\leq}\alpha}(\x{F},\x{V})$, and
$S(l,p) = S(l|_p)$ for all $p \in \VPos(l)$.
\end{definition}

The following lemma is well-known (e.g.~\cite{W92}) and easy
to prove.

\begin{lemma}
If $\RR$ is $S$-compatible then $\x{T}_S(\x{F},\x{V})$
and $\x{T}_{{\leq}\alpha}(\x{F},\x{V})$ for every
$\alpha \in S$ are closed under rewriting by $\RR$.
\qed
\end{lemma}

The following result is a special case of \cite[Theorem 6.2]{FZM11}.

\begin{theorem}
\label{THM:persistence}
An $S$-compatible left-linear TRS $\RR$
is confluent on $\x{T}(\x{F},\x{V})$
if and only if it is confluent on $\x{T}_S(\x{F},\x{V})$.
\qed
\end{theorem}

\begin{example}
\label{EX:pl}
Consider the duplicating TRS $\RR$ with rules
\begin{xalignat*}{2}
1\colon~\m{f}(\m{a}) &\to \m{f}(\m{b}) &
2\colon~\m{f}(x) &\to \m{g}(\m{f}(x),\m{f}(x))
\end{xalignat*}
Recall that L-labelings (in particular, rule labelings) that are not
LL-labelings are not applicable to non-linear TRSs because the variable
overlap diagram (Figure~\ref{FIG:lpeaks}\subref{FIG:vll})
is not decreasing.
Let $S = \{ 0, 1 \}$ with the following sort declarations:
\begin{xalignat*}{5}
x &: 0 &
\m{a} &: 0 &
\m{b} &: 0 &
\m{f} &: 0 \to 1 &
\m{g} &: 1 \times 1 \to 1
\end{xalignat*}
The TRS $\RR$ is $S$-compatible and hence we may
restrict rewriting to $S$-sorted terms
without affecting confluence by Theorem~\ref{THM:persistence}.
This has the beneficial effect that
variable overlaps are ruled out. To see how, note that
no subterms of sort $1$ can appear inside terms of sort $0$.
Consider the left-hand side $\m{f}(x)$ of $\RR$. We have
$S(\m{f}(x),1) = 0$, so that any term substituted for $x$
must have sort $0$. Further note that both left-hand sides
have sort $1$. Consequently, no rule application may be nested below
$\m{f}(x) \to \m{g}(\m{f}(x),\m{f}(x))$
and hence variable overlaps are ruled out.
Therefore, we may use L-labelings
to show confluence of $\RR$ even though $\RR$ is not linear,
and in fact the rule labeling 
which takes the rule numbers as labels
allows us to join the sole (modulo symmetry) critical peak 
${t = \m{f}(\m{b}) \xl[1]{} \m{f}(\m{a}) \xr[2]{}
\m{g}(\m{f}(\m{a}),\m{f}(\m{a})) = u}$ decreasingly:
${t \xr[2]{} \m{g}(\m{f}(\m{b}),\m{f}(\m{b}))
\xl[1]{} \m{g}(\m{f}(\m{b}),\m{f}(\m{a})) \xl[1]{} u}$.
\end{example}

Formally, we define
$\x{T}_{{\trianglelefteq}\alpha}(\x{F},\x{V}) = \{ t \mid
\text{$t \trianglelefteq t'$ for some
$t' \in \x{T}_{{\leq}\alpha}(\x{F},\x{V})$} \}$, to capture
which terms may occur as subterms of terms of sort $\alpha$
or below.

\begin{theorem}
\label{THM:pl}
Let $\RR$ be a left-linear $S$-compatible TRS
such that the variable $l|_p$ occurs at most once in $r$ 
whenever
$l \to r \in \RR$ and $l' \to r' \in \RR$ with
$l' \in \x{T}_{{\trianglelefteq}S(l,p)}(\x{F},\x{V})$ for
some $p \in \VPos(l)$.
Then $\RR$ is confluent if all its critical peaks are L-decreasing.
\end{theorem}

\begin{proof}
By Theorem~\ref{THM:persistence} we may restrict rewriting to
$S$-sorted terms. The proof follows that of Theorem~\ref{THM:l},
except in the analysis of local peaks, where right-linearity of $\RR$
is used, which is not among our assumptions. Instead, we argue as
follows: Since $\RR$ is left-linear, any local peak has the shape
\XP, \XO, or \XVLL. In the
latter case, the step ${s \xr{q, l' \to r'} t}$ is nested below
${s \xr{p, l \to r} u}$, and it is easy to see that this implies
${l' \in \x{T}_{{\trianglelefteq}S(l,q')}(\x{F},\x{V})}$ for
some variable position $q'$ of $l$ such that $pq' \leqslant q$.
Consequently the variable $x = l|_{q'}$ occurs at most once in $r$
by assumption, and the parallel step (which contains one rewrite
step for every occurrence of $x$ in $r$) is empty or a single step,
resulting in a decreasing diagram.
\qed
\end{proof}

As a refinement of Theorem~\ref{THM:pl}, instead of ruling out
duplicating \XVLL overlaps completely, we
can also add additional constraints on the labeling for the
remaining variable overlaps.

\begin{definition}
Let $\ell$ be a weak LL-labeling for an $S$-compatible TRS $\RR$.
We call $\ell$ \emph{persistent} if whenever
rules $l \to r, l' \to r' \in \RR$ satisfy
$l' \in \x{T}_{{\trianglelefteq}S(l,p)}(\x{F},\x{V})$ for
some $p \in \VPos(l)$, either $|r|_{l|_p} \leqslant 1$
or $\beta > \overline\gamma$ in Figure~\ref{FIG:lpeaks}\subref{FIG:vll}
for all resulting variable overlaps with $l' \to r'$ below $l \to r$.
We call $\RR$ persistent LL-decreasing if there is a
persistent, weak LL-labeling $\ell$ such that all critical peaks
of $\RR$ are decreasing with respect to $\ell$.
\end{definition}

\begin{theorem}
\label{THM:pll}
Let $\RR$ be a left-linear TRS. 
If the critical peaks of $\RR$ are persistent LL-decreasing then
$\RR$ is confluent.
\end{theorem}
\begin{proof}
The proof follows along the lines of the proof of Theorem~\ref{THM:pl}.
In the case of a duplicating variable-left-linear overlap, the
additional constraints ensure that the resulting diagram is
decreasing.
\qed
\end{proof}

\begin{example}
\label{EX:pll}
Suppose we extend the TRS from Example~\ref{EX:pl} with the rule
$\m{a} \to \m{b}$, using the same sorts:
\begin{xalignat*}{3}
1\colon~\m{f}(x) &\to \m{g}(\m{f}(x),\m{f}(x)) &
2\colon~\m{f}(\m{a}) &\to \m{f}(\m{b}) &
3\colon~\m{a} &\to \m{b}
\end{xalignat*}
Theorem~\ref{THM:pl} is no longer applicable, because
rule $3$ may be nested below rule $1$, which is duplicating.
However, by the preceding remark, any rule labeling
with $\ell^i_\RL(1) > \ell^i_\RL(3)$ will make the corresponding
variable overlaps decreasing.
\end{example}

\begin{remark}
Note that Theorem~\ref{THM:pll} does not subsume Theorem~\ref{THM:pl},
because the former demands a weak LL-labeling whereas the latter
requires only an L-labeling. If we were to restrict the L-labeling and
weak LL-labeling conditions to those variable overlaps that are consistent
with the sort declarations, then Theorem~\ref{THM:pll} would subsume
Theorem~\ref{THM:pl}. We chose not to do so because all our labelings
are weak LL-labelings.
\end{remark}

The following example shows that considering order-sorted
instead of many-sorted signatures is beneficial.

\begin{example}
Consider the duplicating TRS $\RR$ given by the rules
\begin{xalignat*}{3}
1\colon~\m{h}(\m{a},\m{a}) &\to \m{f}(\m{a}) &
2\colon~\m{f}(\m{a}) &\to \m{a} &
3\colon~\m{f}(x) &\to \m{h}(x,x) 
\end{xalignat*}
Furthermore, let $\x S = \{ 0, 1 \}$ with $1 > 0$ and
take the sort declarations
\begin{xalignat*}{3}
\m{h} &: 0 \times 0 \to 1 &
\m{f} &: 0 \to 1 &
\m{a} &: 0
\end{xalignat*}
Considering only $\x S$-sorted terms, no rule can be nested below the
duplicating rule $\m{f}(x) \to \m{h}(x,x)$.
Basically, there is one critical peak,
$\m h(\m a,\m a) \xl{3} \m f(\m a) \xr{2} \m a$, which is 
decreasingly joinable as $\m h(\m a,\m a) \xr{1} \m f(\m a) \xr{2} \m a$
by the rule labeling (using rule numbers as labels), and confluence follows
by Theorem~\ref{THM:pl}. Due to the rule $\m{f}(\m{a}) \to \m{a}$, any
many-sorted sort declaration for $\RR$ must assign the same
sorts to $\m{a}$ and the argument and result types of $\m{f}$.
Therefore, $\m{f}(x) \to \m{h}(x,x)$ may be nested below itself,
and Theorems~\ref{THM:pl} and~\ref{THM:pll} would fail in connection
with the rule labeling.
\end{example}

\section{Labelings for Parallel Rewriting}
\label{LAB:par:main}

In this section, rather than labeling individual rewrite steps, we
will label parallel rewrite steps instead. This is inspired by the
parallel moves lemma, which says that any peak $t \parfrom s \parto u$
of two non-overlapping parallel rewrite steps can be joined in a
diamond as $t \parto {\cdot} \parfrom u$, and diamonds are
comparatively easy to label decreasingly, as we saw in
Section~\ref{CR:l}.

The main problem is to label parallel steps such that variable
overlaps are decreasing. The multiset
of the single steps' labels does not work since 
$\{\alpha \} \not\geqslant_\mul 
\{\alpha, \dots, \alpha \}$.
Hence we use sets to label parallel steps which we denote 
by capital Greek
letters. Sets of labels are ordered by the Hoare preorder
of $({\geqslant},{>})$, which we denote by $({\geqslant_H},{>_H})$ 
and is defined by
\begin{alignat*}{1}
\Gamma >_H \Delta
&\quad\iff\quad \Gamma \neq \varnothing \wedge
\forall \beta \in \Delta\:\exists\,\alpha \in \Gamma\:
(\alpha > \beta)
\\
\Gamma \geqslant_H \Delta
&\quad\iff\quad
\forall \beta \in \Delta\:\exists\,\alpha \in \Gamma\:
(\alpha \geqslant \beta)
\end{alignat*}
For readability we drop the subscript $H$ when attaching labels
to rewrite steps as in $\xR[\Vee\Gamma]{}$.

\begin{example}
Let $\geqslant$ denote the natural order on $\Nat$.
Then $\{ 1 \} \geqslant_H \{ 0 , 1 \}$ and
$\{ 1 \} \geqslant_H \{ 1, 1, 1 \} = \{ 1 \}$ but
$\{ 5, 4 \} \not>_H \{ 5, 3 \}$.
\end{example}

The following lemma states obvious properties of Hoare preorders which
we implicitly use in the sequel.

\begin{lemma}
Let $(\geqslant_H,>_H)$ be a Hoare preorder.
\begin{enumerate}
\item
If $(\geqslant,>)$ is a monotone reduction pair then
$(\geqslant_H,>_H)$ is a monotone reduction pair.
\item
If $\Gamma \supseteq \Gamma'$ then $\Gamma \geqslant_H \Gamma'$.
\item
If $\Gamma >_H \Gamma'$ and $\Delta >_H \Delta'$ then
$\Gamma \cup \Delta >_H \Gamma' \cup \Delta'$.
\item
If $\Gamma \geqslant_H \Gamma'$ and $\Delta \geqslant_H \Delta'$ then
$\Gamma \cup \Delta \geqslant_H \Gamma' \cup \Delta'$.
\qed
\end{enumerate}
\end{lemma}

As we have seen in Section~\ref{CR:ll}, constructing LL-labelings
is quite a bit harder than constructing L-labelings, because of the
duplicated steps in the \XVLL case
(Figure~\ref{FIG:lpeaks}\subref{FIG:vll}). 
Here, we use weak LL-labelings for labeling single and parallel
rewrite steps.
Throughout this section we assume a given left-linear
TRS $\RR$, and a weak LL-labeling $\lab$ with corresponding labeling
function for parallel steps $\plab$, as introduced in the following
definition.

\begin{definition}
\label{DEF:plab}
We lift a weak LL-labeling $\ell$ to parallel steps $t \xR{P} t'$
as follows. For each $\pi \in P$, we have
a rewrite step $t \xr{\pi} t^\pi$. We label $t \xR{P} t'$ by
$\plab(t \xR{P} t') = \{ \lab(t \xr{\pi} t^\pi) \mid \pi \in P \}$.
\end{definition}

So a parallel rewrite step is labeled by the set of
the labels of the single steps making up the parallel step. We
indicate labels along with the step, writing $t \xR[\Gamma]{P} t'$.

The next example shows that the labels change when decomposing a
parallel step into 
a sequence of single steps, i.e., the label of the parallel
step may be different from the union of labels of the single steps.
However, the proof of Lemma~\ref{LEM:pmanip} reveals that for weak
LL-labelings the labels never increase when 
sequencing a parallel step.

\begin{example}
Consider the rule $\m{a} \to \m{b}$ and the extension of the
source labeling $\lab(s \to t) = s$ to parallel steps.
Then $\m{f}(\m{a},\m{a}) \xR[\{ \m{f}(\m{a},\m{a}) \}]{}
\m{f}(\m{b},\m{b})$ but $\m{f}(\m{a},\m{a})
\xR[\{ \m{f}(\m{a},\m{a}) \}]{} \m{f}(\m{b},\m{a})
\xR[\{ \m{f}(\m{b},\m{a}) \}]{} \m{f}(\m{b},\m{b})$.
Clearly $\{ \m{f}(\m{a},\m{a}) \} \neq
\{ \m{f}(\m{a},\m{a}), \m{f}(\m{b},\m{a}) \}$.
This effect is intrinsic to labelings that take the context of the 
rewrite step into account. On the other hand, the rule labeling
gives
$\m{f}(\m{a},\m{a}) \xR[\{ 1 \}]{} \m{f}(\m{b},\m{b})$ and
$\m{f}(\m{a},\m{a}) \xR[\{ 1 \}]{} \m{f}(\m{b},\m{a})
\xR[\{ 1 \}]{} \m{f}(\m{b},\m{b})$ with $\{ 1 \} = \{ 1, 1 \}$,
because the labels are independent of the context.
\end{example}

\begin{figure}
\centering
\begin{tikzpicture}[]
\node (s) at (0,0) {$s$};
\node (sp) at (-2,-0.5) {$s^\pi$};
\node (sP) at (-1,-2.5) {$s^{P'}$};
\node (t1) at (-3,-3) {$t_1$};
\node (t2) at (4,-3) {$t_2$};
\node (t2p) at (2,-3.5) {$t_2^\pi$};
\node (t2P) at (3,-5.5) {$t_2^{P'}$};
\node (u) at (1,-6) {$u$};
\draw[To] (s) -- node[sloped,above]{$\scriptstyle Q$}
node[sloped,below]{$\scriptstyle \Delta$} (t2);
\draw[To] (s) -- node[sloped,above]{$\scriptstyle P$}
 node[sloped,below]{$\scriptstyle \Gamma$} (t1);
\draw[To] (s) -- node[sloped,above]{$\scriptstyle P'$}
 node[sloped,below]{$\scriptstyle \Veq\Gamma$} (sP);
\draw[to] (s) -- node[sloped,above]{$\scriptstyle \pi$}
 node[sloped,below]{$\scriptstyle \Veq\Gamma$} (sp);
\draw[to,densely dashed] (sP) -- (t1);
\draw[To] (sp) -- node[sloped,above]{$\scriptstyle P'$}
 node[sloped,below]{$\scriptstyle \Veq\Gamma$} (t1);
\draw[To] (t1) -- node[sloped,above]{$\scriptstyle Q$}
 node[sloped,below]{$\scriptstyle \Veq\Delta$} (u);
\draw[To] (t2) -- node[sloped,above]{$\scriptstyle P$}
 node[sloped,below]{$\scriptstyle \Veq\Gamma$} (u);
\draw[to] (t2) -- node[sloped,above]{$\scriptstyle \pi$}
 node[sloped,below]{$\scriptstyle \Veq\Gamma$} (t2p);
\draw[To] (t2) -- node[sloped,above]{$\scriptstyle P'$}
 node[sloped,below]{$\scriptstyle \Veq\Gamma$} (t2P);
\draw[To,densely dashed] (t2p) -- (u);
\draw[to,densely dashed] (t2P) -- (u);
\draw[To] (sp) -- node[sloped,above]{$\scriptstyle Q$}
 node[sloped,below]{$\scriptstyle \Veq\Delta$} (t2p);
\draw[To,densely dashed] (sP) -- (t2P);
\end{tikzpicture}
\caption{Weak LL-labeling applied to parallel steps.}
\label{FIG:pmanip1}
\end{figure}
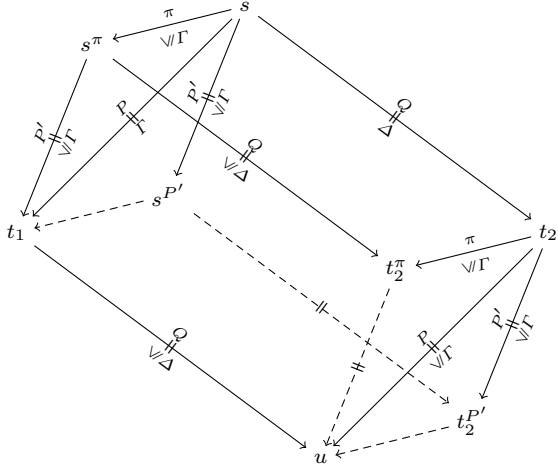

The following lemma is the key to show that even for parallel rewriting
overlaps due to Figure~\ref{FIG:lpeaks}\subref{FIG:p} (parallel) and 
Figure~\ref{FIG:lpeaks}\subref{FIG:vll} \XVLL
are decreasing.

\begin{lemma}
\label{LEM:pmanip}
\mbox{}
\begin{enumerate}
\item 
\label{LEM:pmanip1}
Let $t_1 \xL[\Gamma]{P} s \xR[\Delta]{Q} t_2$ with $P \parallel Q$.
Then there is a term $u$ such that $s \xR[\Gamma \cup \Delta]{P \cup Q} u$
and $t_1 \xR[\Delta']{Q} u \xL[\Gamma']{P} t_2$, where
$\Gamma \geqslant_H \Gamma'$ and $\Delta \geqslant_H \Delta'$.
\item 
\label{LEM:pmanip2}
Let $s \xR{} s'$ and $\sigma(x) \xR{} \sigma'(x)$ for all $x \in \x{V}$,
so that there are parallel rewrite steps
$s\sigma' \xL[\Gamma]{P} s\sigma \xR[\Delta]{Q} s'\sigma$.
Then $s\sigma' \xR[\Delta']{Q} s'\sigma' \xL[\Gamma']{} s'\sigma$
and $\Gamma \geqslant_H \Gamma'$, ${\Delta \geqslant_H \Delta'}$.
Furthermore, if $\sigma(x) = \sigma'(x)$ for all $x \in \Var(s'|_Q)$
then $s\sigma \xR[\Sigma]{} s'\sigma'$ for some
$\Sigma \subseteq \Gamma \cup \Delta$.
\end{enumerate}
\end{lemma}
\begin{proof}
{\itshape\ref{LEM:pmanip1}}.
First note that since $P \parallel Q$, a term $u$ with $s \xR{P\cup Q} u$
exists. We have
\begin{align*}
\plab(s \xR{P \cup Q} u)
&= \{ \lab(s \xr{\pi} s^\pi) \mid \pi \in P \cup Q\} \\
&= \{ \lab(s \xr{\pi} s^\pi) \mid \pi \in P \}
\cup \{ \lab(s \xr{\pi} s^\pi) \mid \pi \in Q \} \\
&= \plab(s \xR{P} t_1) \cup \plab(s \xR{Q} t_2) = \Gamma \cup \Delta
\end{align*}
by definition.
To establish $t_1 \xR[\Delta']{Q} u \xL[\Gamma']{P} t_2$, we
use induction on $|P|+|Q|$. We consider several base
cases.  If $|P| = 0$ or $|Q| = 0$ then the result follows by definition of 
parallel rewriting. If $|P| = |Q| = 1$ the result follows
from the fact that $\ell$ is a weak LL-labeling,
Definition~\ref{DEF:ll}(\ref{DEF:ll1})
(Figure~\ref{FIG:lpeaks}\subref{FIG:p}).
For the induction step, assume without loss of generality that $|P| > 1$
and let $P = \{ \pi \} \uplus P'$.
The proof is illustrated in Figure~\ref{FIG:pmanip1}.
The parallel $P$-step can be decomposed into a $\pi$-step and a $P'$-step.
Since $\{ \pi \}, P' \subseteq P$, the labels are less than or equal
to $\Gamma$.
Then we apply the induction hypothesis to the peaks
\begin{enumerate}
\renewcommand{\theenumi}{\roman{enumi}}%
\item
$s^{P'} \xL[\Veq\Gamma]{P'} s \xR[\Veq\Gamma]{\{ \pi \}} s^\pi$
yielding $s^\pi \xR[\Veq\Gamma]{P'} t_1$,
\item
$s^\pi \xL[\Veq\Gamma]{\{ \pi \}} s \xR[\Delta]{Q} t_2$
yielding $t_2 \xR[\Veq\Gamma]{\{ \pi \}} t_2^\pi$
and $s^\pi \xR[\Veq\Delta]{Q} t_2^\pi$,
\item
$s^{P'} \xL[\Veq\Gamma]{P'} s \xR[\Delta]{Q} t_2$
yielding $t_2 \xR[\Veq\Gamma]{P'} t_2^{P'}$,
which we merge with
$t_2 \xR[\Veq\Gamma]{\{ \pi \}} t_2^\pi$
to obtain $t_2 \xR[\Veq\Gamma]{P} u$, noting that the union of
two sets from $\Veq\Gamma$ is again in $\Veq\Gamma$, and finally
\item
$t_1 \xL[\Veq\Gamma]{P'} s^\pi \xR[\Veq\Delta]{Q} t_2^\pi$
yielding $t_1 \xR[\Veq\Delta]{Q} u$.
\end{enumerate}

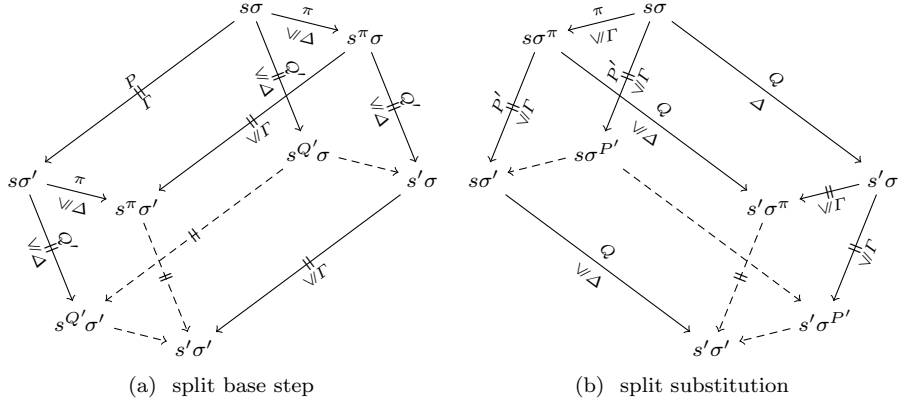
\begin{figure}[t]
\subfloat[\label{FIG:pmanip2:r} split base step]{
\begin{tikzpicture}[scale=0.75]
\node (s) at (0,0) {$s\sigma$};
\node (sp) at (2,-0.5) {$s^\pi\sigma$};
\node (sP) at (1,-2.5) {$s^{Q'}\!\sigma$};
\node (t1) at (3,-3) {$s'\sigma$};
\node (t2) at (-4,-3) {$s\sigma'$};
\node (t2p) at (-2,-3.5) {$s^\pi\sigma'$};
\node (t2P) at (-3,-5.5) {$s^{Q'}\!\sigma'$};
\node (u) at (-1,-6) {$s'\sigma'$};
\draw[To] (s) -- node[sloped,above]{$\scriptstyle P$}
node[sloped,below]{$\scriptstyle \Gamma$} (t2);
\draw[To] (s) -- node[sloped,above]{$\scriptstyle Q'$}
 node[sloped,below]{$\scriptstyle \Veq\Delta$} (sP);
\draw[to] (s) -- node[sloped,above]{$\scriptstyle \pi$}
 node[sloped,below]{$\scriptstyle \Veq\Delta$} (sp);
\draw[to,densely dashed] (sP) -- (t1);
\draw[To] (sp) -- node[sloped,above]{$\scriptstyle Q'$}
 node[sloped,below]{$\scriptstyle \Veq\Delta$} (t1);
\draw[To] (t1) -- node[sloped,below]{$\scriptstyle \Veq\Gamma$} (u);
\draw[to] (t2) -- node[sloped,above]{$\scriptstyle \pi$}
 node[sloped,below]{$\scriptstyle \Veq\Delta$} (t2p);
\draw[To] (t2) -- node[sloped,above]{$\scriptstyle Q'$}
 node[sloped,below]{$\scriptstyle \Veq\Delta$} (t2P);
\draw[To,densely dashed] (t2p) -- (u);
\draw[to,densely dashed] (t2P) -- (u);
\draw[To] (sp) -- node[sloped,below]{$\scriptstyle \Veq\Gamma$} (t2p);
\draw[To,densely dashed] (sP) -- (t2P);
\end{tikzpicture}
}
\hfill
\subfloat[\label{FIG:pmanip2:l} split substitution]{
\begin{tikzpicture}[scale=0.75]
\node (s) at (0,0) {$s\sigma$};
\node (sp) at (-2,-0.5) {$s\sigma^\pi$};
\node (sP) at (-1,-2.5) {$s\sigma^{P'}$};
\node (t1) at (-3,-3) {$s\sigma'$};
\node (t2) at (4,-3) {$s'\sigma$};
\node (t2p) at (2,-3.5) {$s'\sigma^\pi$};
\node (t2P) at (3,-5.5) {$s'\sigma^{P'}$};
\node (u) at (1,-6) {$s'\sigma'$};
\draw[to] (s) -- node[sloped,above]{$\scriptstyle Q$}
node[sloped,below]{$\scriptstyle \Delta$} (t2);
\draw[To] (s) -- node[sloped,above]{$\scriptstyle P'$}
 node[sloped,below]{$\scriptstyle \Veq\Gamma$} (sP);
\draw[to] (s) -- node[sloped,above]{$\scriptstyle \pi$}
 node[sloped,below]{$\scriptstyle \Veq\Gamma$} (sp);
\draw[to,densely dashed] (sP) -- (t1);
\draw[To] (sp) -- node[sloped,above]{$\scriptstyle P'$}
 node[sloped,below]{$\scriptstyle \Veq\Gamma$} (t1);
\draw[to] (t1) -- node[sloped,above]{$\scriptstyle Q$}
 node[sloped,below]{$\scriptstyle \Veq\Delta$} (u);
\draw[To] (t2) -- node[sloped,below]{$\scriptstyle \Veq\Gamma$} (t2p);
\draw[To] (t2) -- node[sloped,below]{$\scriptstyle \Veq\Gamma$} (t2P);
\draw[To,densely dashed] (t2p) -- (u);
\draw[to,densely dashed] (t2P) -- (u);
\draw[to] (sp) -- node[sloped,above]{$\scriptstyle Q$}
 node[sloped,below]{$\scriptstyle \Veq\Delta$} (t2p);
\draw[to,densely dashed] (sP) -- (t2P);
\end{tikzpicture}
}
\caption{Weak LL-labeling applied to nested parallel steps.}
\label{FIG:pmanip2}
\end{figure}

{\itshape\ref{LEM:pmanip2}}.
The existence of parallel rewrite steps $s\sigma' \xR{} s'\sigma'$
and $s'\sigma \xR{} s'\sigma'$ follows easily from the definition of
parallel steps.
We establish $\Gamma \geqslant_H \Gamma'$ and $\Delta \geqslant_H \Delta'$
by induction on $|Q|$. The reasoning for the induction step ($|Q| > 1$) is
very similar to the induction step in item~\ref{LEM:pmanip1},
cf.~Figure~\ref{FIG:pmanip2}\subref{FIG:pmanip2:r}:
Taking $Q = \{ \pi \} \uplus Q'$, we split
$s\sigma \xR[\Delta]{Q} s'\sigma$ into
$s\sigma \xR[\Veq\Delta]{\{\pi\}} s^\pi\sigma$ and
$s\sigma \xR[\Veq\Delta]{Q'} s^{Q'}\!\sigma$.
We apply the induction hypothesis to the peaks
\begin{enumerate}
\renewcommand{\theenumi}{\roman{enumi}}%
\item
$s\sigma' \xL[\Gamma]{P} s\sigma \xR[\Veq\Delta]{\{\pi\}} s^\pi\sigma$
yielding $s\sigma' \xR[\Veq\Delta]{\{\pi\}} s^\pi\sigma'$ and
$s^\pi\sigma \xR[\Veq\Gamma]{} s^\pi\sigma'$,
\item
$s\sigma' \xL[\Gamma]{P} s\sigma \xR[\Veq\Delta]{Q'} s^{Q'}\sigma$
yielding $s\sigma' \xR[\Veq\Delta]{Q'} s^{Q'}\!\sigma'$, which
can be merged with $s\sigma' \xR[\Veq\Delta]{\{\pi\}} s^\pi\sigma'$
to obtain $s\sigma' \xR[\Veq\Delta]{Q} s'\sigma'$, and finally
\item
$s^\pi\sigma' \xL[\Veq\Gamma]{} s^\pi\sigma \xR[\Veq\Delta]{Q'} s'\sigma$
yielding $s'\sigma \xR[\Veq\Gamma]{} s'\sigma'$,
where $s^\pi\sigma \xR[\Veq\Delta]{Q'} s'\sigma$
is obtained from part~\ref{LEM:pmanip1} of this lemma
applied to $s^{Q'}\!\sigma \xL[\Veq\Delta]{Q'} s\sigma
\xR[\Veq\Delta]{\{\pi\}} s^\pi\sigma$.
\end{enumerate}
This concludes the induction step.
If $|Q| = 0$, there is nothing to show, so only the base case $|Q| = 1$
remains.
Note that because $\RR$ is left-linear, we may assume without
loss of generality that $s$ is linear. Therefore, every rewrite step
of $s\sigma \xR{P} s\sigma'$ can be performed by modifying $\sigma$.
For $P' \subseteq P$, we write $\sigma^{P'}$ for the substitution
$\tau$ that satisfies $s\sigma \xR{P'} s\tau$,
and proceed by induction on $|P|$.
For the induction step ($|P| > 1$), the argument is again almost the
same as before, cf.\ Figure~\ref{FIG:pmanip2}\subref{FIG:pmanip2:l}.
Let $P = \{ \pi \} \uplus P'$.
We split $s\sigma \xR[\Gamma]{P} s\sigma'$ into
$s\sigma \xR[\Veq\Gamma]{\{\pi\}} s\sigma^\pi$ and
$s\sigma \xR[\Veq\Gamma]{P'} s\sigma^{P'}$.
Next we apply the induction hypothesis to the peaks
\begin{enumerate}
\renewcommand{\theenumi}{\roman{enumi}}%
\item
$s\sigma^\pi \xL[\Veq\Gamma]{\{\pi\}} s\sigma \xr[\Delta]{Q} s'\sigma$
yielding $s\sigma^\pi \xr[\Veq\Delta]{Q} s'\sigma^\pi$ and
$s'\sigma \xR[\Veq\Gamma]{} s'\sigma^\pi$,
\item
$s\sigma^{P'} \xL[\Veq\Gamma]{P'} s\sigma \xr[\Delta]{Q} s'\sigma$
yielding $s'\sigma \xR[\Veq\Gamma]{} s'\sigma^{P'}$, which
can be merged with $s'\sigma \xR[\Veq\Gamma]{} s'\sigma^\pi$
to obtain $s'\sigma \xR[\Veq\Gamma]{} s'\sigma'$, and finally
\item
$s\sigma' \xL[\Veq\Gamma]{P'} s\sigma^\pi \xr[\Veq\Delta]{Q} s'\sigma^\pi$
yielding $s\sigma' \xr[\Veq\Delta]{Q} s'\sigma'$,
where $s\sigma^\pi \xR[\Veq\Gamma]{P'} s\sigma'$
is obtained from part~\ref{LEM:pmanip1} of this lemma
applied to $s\sigma^\pi \xL[\Veq\Gamma]{\{\pi\}} s\sigma
\xR[\Veq\Gamma]{P'} s\sigma^{P'}$.
\end{enumerate}
This concludes the induction step.
If $|P| = 0$ then there is nothing to show. Finally, if $|P| = |Q| = 1$,
then we are left with a parallel or variable overlap, and we conclude
by Definition~\ref{DEF:wll}(\ref{DEF:ll1}) or \ref{DEF:wll}(\ref{DEF:ll2}),
respectively. This concludes the proof that $\Gamma \geqslant_H \Gamma'$
and $\Delta \geqslant_H \Delta'$.
Now if $\sigma(x) = \sigma'(x)$ for all $x \in \Var(s'|_Q)$,
then $s'\sigma \xR{P'} s'\sigma'$ satisfies $P' \parallel Q$.
Performing the same rewrite steps on $s\sigma$, we obtain a
parallel rewrite step $s\sigma \xR{P'} s''$ with
$P' \subseteq P$ and therefore
${\Gamma'' = \plab(s\sigma \xR{P'} s'') \subseteq
\plab(s\sigma \xR{P} s\sigma') = \Gamma}$.
Finally, using the first part of this lemma, we can combine the
two parallel steps from $s\sigma$ into a single one,
$s\sigma \xR[\Gamma'' \cup \Delta]{P' \cup Q} s'\sigma'$
with
$\Sigma = \Gamma'' \cup \Delta \subseteq \Gamma \cup \Delta$
as claimed.
\qed
\end{proof}

Only Definition~\ref{DEF:ll}(\ref{DEF:ll1}) was used in the proof
of Lemma~\ref{LEM:pmanip}(\ref{LEM:pmanip1}). This fact
can be exploited for an alternative characterization of weak LL-labelings.

\begin{corollary}
Let $\ell$ be a labeling. Then $\ell$ is a weak LL-labeling if and
only if
\begin{enumerate}
\item
in Figure~\ref{FIG:lpeaks}\subref{FIG:p}, $\alpha \geqslant \gamma$
and $\beta \geqslant \delta$, and
\item
in Figure~\ref{FIG:lpeaks}\subref{FIG:vll},
$\beta \geqslant \delta$ and $\{ \alpha \} \geqslant_H \plab(u \xR{} v)$.
\end{enumerate}
\end{corollary}

\begin{proof}
Assume that $\ell$ is a weak LL-labeling. The first condition of
this lemma is identical to Definition~\ref{DEF:ll}(\ref{DEF:ll1}).
For the second condition, $\beta \geqslant \delta$ follows from
Definition~\ref{DEF:ll}(\ref{DEF:ll2}). To establish
$\{ \alpha \} \geqslant_H \plab(u \xR{P} v)$, we need to show that
$\alpha \geqslant \lab(u \xR{\pi} u^\pi)$ for all $\pi \in P$. For
each $\pi$, we can arrange that $\lab(u \xR{\pi} u^\pi) = \gamma_1$
by choosing $u \xR{\pi} u^\pi$ as the first step in the permutation
of $u \xR{} v$, and then $\alpha \geqslant \gamma_1$ follows from
Definition~\ref{DEF:ll}(\ref{DEF:ll2}), establishing the claim.

Next assume that $\ell$ satisfies the conditions of this lemma.
Then the condition of Definition~\ref{DEF:ll}(\ref{DEF:ll1}) holds.
To show the conditions of Definition~\ref{DEF:ll}(\ref{DEF:ll2}), note that
$\beta \geqslant \delta$ holds by assumption. Consider the parallel
rewrite step $u \xR{P} v$ and a permutation $\seq{\pi}$ of $P$.
We can decompose $u \xR{P} v$ into
$u = u_0 \xr[\gamma_1]{\pi_1} u_1
\xr[\gamma_2]{\pi_2} \cdots
\xr[\gamma_n]{\pi_n} u_n = v.$
By Lemma~\ref{LEM:pmanip}(\ref{LEM:pmanip1}) applied to the peaks
\[
{\cdot} \xL[\Veq\{ \alpha \}]{\{ \pi_i \}} u
\xR[\Veq\{ \alpha \}]{\{ \pi_1, \dots, \pi_{i-1} \}} u_{i-1}
\]
we obtain
$u_{i-1} \xR[\Veq\{\alpha\}]{\{ \pi_i \}} u_i$, i.e.,
$\{ \alpha \} \geqslant_H \{ \gamma_i \}$, which is
equivalent to $\alpha \geqslant \gamma_i$.
Hence $\alpha \geqslant \overline{\gamma}$.
\qed
\end{proof}

The following lemma is used to reduce the number of parallel peaks that
have to be considered in the proof of
Theorem~\ref{thm-main-basic}.

\begin{lemma}
\label{LEM:pmerge}
Let $s \xR[\Vee\Gamma]{*} \cdot \xR[\Veq\Delta]{{}} \cdot
\xR[\Vee\Gamma\Delta]{*} t$ and $s \xR[\Vee\Gamma]{*} \cdot
\xR[\Veq\Delta]{{}} \cdot \xR[\Vee\Gamma\Delta]{*} u$
be two rewrite sequences
such that all rewrite steps in the sequence to $t$ are at or below
a position
$p$ and the rewrite steps in the sequence to $u$ are parallel to
$p$. Then the two rewrite sequences can be merged into
$s \xR[\Vee\Gamma]{*} {\cdot} \xR[\Veq\Delta]{{}} {\cdot}
\xR[\Vee\Gamma\Delta]{*} u[t|_p]_p$.
\end{lemma}
\begin{proof}
Let the two sequences be
$s \xR[\Vee\Gamma]{*} t_1 \xR[\Veq\Delta]{{}} t_2
\xR[\Vee\Gamma\Delta]{*} t$ and
$s \xR[\Vee\Gamma]{*} u_1 \xR[\Veq\Delta]{{}} u_2
\xR[\Vee\Gamma\Delta]{*} u$. Using
Lemma~\ref{LEM:pmanip}(\ref{LEM:pmanip1}) repeatedly, we can derive a
sequence
\[
s \xR[\Vee\Gamma]{*}
u_1[t_1|_p]_p \xR[\Veq\Delta]{{}} u_2[t_2|_p]_p
\xR[\Vee\Gamma\Delta]{*}
u[t|_p]_p
\]
which establishes the claim.
\qed
\end{proof}

In order to perform a critical pair analysis for parallel rewrite steps, we
need parallel critical pairs~\cite{T81,G96}.

\begin{definition}
Let $l \to r$ be a rule in a TRS $\RR$ and $P$ be a non-empty set of
pairwise parallel redex patterns such that every $\pi \in P$
critically overlaps with $l$.
By choosing variants of rules from $\RR$ appropriately, we may assume
that the sets $\Var(l_\pi)$ for $\pi \in P$ and $\Var(l)$ are pairwise
disjoint. Assume that the unification problem
$\{ l|_{p_\pi} \approx l_\pi \mid \pi \in P \}$
has a solution and let
$\sigma$ be a most general unifier. Then there is a unique term $l_P$
such that $l\sigma \xR{P} l_P$. We call
$l_P \pcp r\sigma$ a \emph{parallel critical pair}, and
$l_P \xL{} l\sigma \xr{} r\sigma$ a \emph{parallel critical peak}.
\end{definition}

Note that every standard critical pair also is a parallel critical pair.
The following lemma states how critical pair analysis for a peak
consisting of a parallel and a root rewrite step is done. It is
a straightforward extension of \cite[Lemma 4.7]{G96}.

\begin{lemma}
\label{LEM:pcp}
Let $\RR$ be a left-linear TRS and $t \xL{P} s \xr{\pi} u$ with
$p_\pi = \epsilon$. Then either $P \perp \pi$ or there are substitutions
$\sigma \parto \sigma'$ and a parallel critical pair $t' \pcp u'$ such
that
$t = t'\sigma' \xL{P \setminus P'} t'\sigma \xL{P'} s \to u'\sigma = u$
with $P' \subseteq P$.
\qed
\end{lemma}

Note that left-linearity is essential for the substitutions $\sigma$
and $\sigma'$ to exist in Lemma~\ref{LEM:pcp}.
We are now ready to state and prove the main theorem of this section.

\begin{theorem}
\label{thm-main-basic}
A left-linear TRS $\RR$ is confluent if all its parallel
critical peaks $t \xL[\Gamma]{P} s \xr[\Delta]{} u$
can be joined decreasingly as
\[
t \xr[\Vee\Gamma]{*}
{\cdot}\xR[\Veq\Delta]{}
{\cdot}\xr[\Vee\Gamma\Delta]{*}
{\cdot}\xl[\Vee\Gamma\Delta]{*}
v \xL[\Veq\Gamma]{Q}
{\cdot}\xl[\Vee\Delta]{*}
u
\]
such that $\Var(v|_Q) \subseteq \Var(s|_P)$.
\end{theorem}

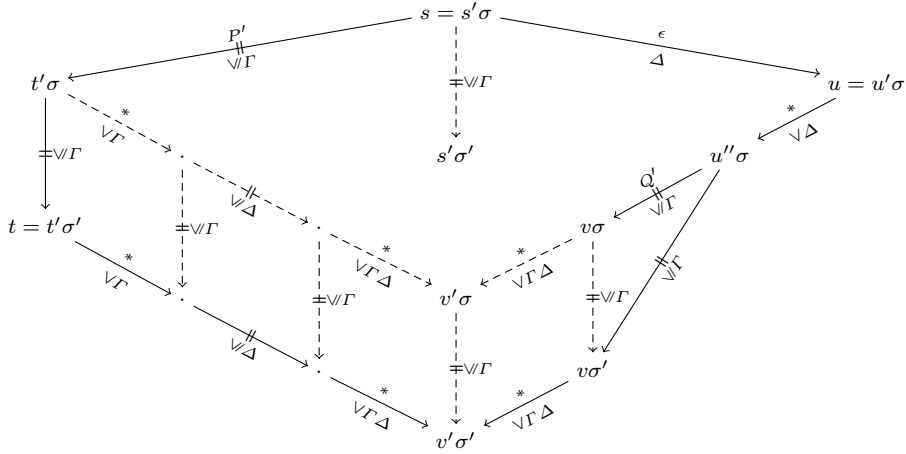
\begin{figure}
\centerline{
\begin{tikzpicture}
\matrix (D) [matrix of math nodes,column sep={1.8cm,between origins},
row sep={.95cm,between origins},ampersand replacement=\&]{
\& \& \& s = s' \sigma \& \& \& \\
t' \sigma \& \& \& \& \& \& u = u' \sigma \\
\& \cdot \& \& s'\sigma' \& \& u'' \sigma \& \\
t = t' \sigma' \& \& \cdot \& \& v\sigma \\ 
\& \cdot \& \& v' \sigma \& \& \& \\
\& \& \cdot \& \& v\sigma' \& \& \\
\& \& \& v' \sigma' \& \& \&\\
};
\draw[To] (D-1-4) -- node[sloped,above] {$\scriptstyle P'$}
node [sloped,below] {$\scriptstyle\Veq\Gamma$} (D-2-1);
\draw[To,densely dashed] (D-1-4) -- node[right]
 {$\scriptstyle\Veq\Gamma$} (D-3-4);
\draw[to] (D-1-4) -- node [sloped,below] {$\scriptstyle\Delta$}
 node [sloped,above] {$\scriptstyle\epsilon$} (D-2-7);
\draw[to,densely dashed] (D-2-1) -- node [sloped,above] {$\scriptstyle *$}
 node [sloped,below] {$\scriptstyle\Vee\Gamma$} (D-3-2);
\draw[To,densely dashed] (D-3-2) -- node [sloped,below]
 {$\scriptstyle\Veq\Delta$} (D-4-3);
\draw[to,densely dashed] (D-4-3) -- node [sloped,above] {$\scriptstyle *$}
 node [sloped,below] {$\scriptstyle\Vee\Gamma\Delta$} (D-5-4);
\draw[to,densely dashed] (D-4-5) -- node [sloped,above] {$\scriptstyle *$}
 node [sloped,below] {$\scriptstyle\Vee\Gamma\Delta$} (D-5-4);
\draw[To] (D-3-6) -- node [sloped,above] {$\scriptstyle Q'$}
 node [sloped,below] {$\scriptstyle\Veq\Gamma$} (D-4-5);
\draw[to] (D-2-7) -- node [sloped,above] {$\scriptstyle *$}
 node [sloped,below] {$\scriptstyle\Vee\Delta$} (D-3-6);
\draw[to] (D-4-1) -- node [sloped,above] {$\scriptstyle *$}
 node [sloped,below] {$\scriptstyle\Vee\Gamma$} (D-5-2);
\draw[To] (D-5-2) -- node [sloped,below]
 {$\scriptstyle\Veq\Delta$} (D-6-3);
\draw[to] (D-6-3) -- node [sloped,above] {$\scriptstyle *$}
 node [sloped,below] {$\scriptstyle\Vee\Gamma\Delta$} (D-7-4);
\draw[to] (D-6-5) -- node [sloped,above] {$\scriptstyle *$}
 node [sloped,below] {$\scriptstyle\Vee\Gamma\Delta$} (D-7-4);
\draw[To] (D-2-1) -- node [right] {$\scriptstyle\Veq\Gamma$} (D-4-1);
\draw[To,densely dashed] (D-3-2) -- node [right]
 {$\scriptstyle\Veq\Gamma$} (D-5-2);
\draw[To,densely dashed] (D-4-3) -- node [right]
 {$\scriptstyle\Veq\Gamma$} (D-6-3);
\draw[To,densely dashed] (D-5-4) -- node [right]
 {$\scriptstyle\Veq\Gamma$} (D-7-4);
\draw[To,densely dashed] (D-4-5) -- node [right]
 {$\scriptstyle\Veq\Gamma$} (D-6-5);
\draw[To] (D-3-6) to node [sloped,below] {$\scriptstyle\Veq\Gamma$}
 (D-6-5);
\end{tikzpicture}
}
\caption{Part of the proof of Theorem~\ref{thm-main-basic}.}
\label{fig-proof-basic}
\end{figure}

\begin{proof}
We show that $\parto$ is decreasing, which implies
confluence of $\RR$.
Let $t \xL[\Gamma]{P} s \xR[\Delta]{Q} u$.
It suffices to show that
\begin{equation}
\label{joinxxx}
t \xr[\Vee\Gamma]{*}
\cdot \xR[\Veq\Delta]{}
\cdot \xr[\Vee\Gamma\Delta]{*}
\cdot \xl[\Vee\Gamma\Delta]{*}
\cdot \xL[\Veq\Gamma]{}
\cdot \xl[\Vee\Delta]{*} u
\end{equation}
Below we show that \eqref{joinxxx} holds whenever $P = \{ \pi \}$ or
$Q = \{ \pi \}$ with $p_\pi = \epsilon$. Then for all
$p \in \min\,\{ p_\pi \mid \pi \in P \cup Q \}$,
$t \xL[\Gamma]{P} s \xR[\Delta]{Q} u$ induces a peak
$t|_p \xL[\Gamma_0]{P'} s|_p \xR[\Delta_0]{Q'} u|_p$,
where $P' = \{ \pi \}$ or $Q' = \{ \pi \}$ for some $\pi$ with
$p_\pi = \epsilon$.
So for each $p$, we obtain a joining sequence for $t|_p$ and $u|_p$
of shape \eqref{joinxxx}. By the monotonicity of labelings,
this results in joining sequences
\[
s[t|_p]_p \xr[\Vee\Gamma]{*}
\cdot \xR[\Veq\Delta]{}
\cdot \xr[\Vee\Gamma\Delta]{*}
\cdot \xl[\Vee\Gamma\Delta]{*}
\cdot \xL[\Veq\Gamma]{}
\cdot \xl[\Vee\Delta]{*} s[u|_p]_p
\]
which are mutually parallel since the positions $p \in \min (P \cup Q)$
are mutually parallel. By repeated application of Lemma~\ref{LEM:pmerge}
those sequences can be combined into a single sequence of the same shape.

In order to show \eqref{joinxxx} for $P = \{ \pi \}$ or
$Q = \{ \pi \}$ with $p_\pi = \epsilon$, assume without loss of generality
that $Q = \{ \pi \}$.
If $P \perp \pi$ then $s = l_\pi \sigma$ and,
because $l_\pi$ is linear, there is a substitution $\sigma'$
with $t = l_\pi \sigma'$ and $\sigma(x) \xR{} \sigma'(x)$ for all variables
$x \in \x{V}$. We conclude by Lemma~\ref{LEM:pmanip}(\ref{LEM:pmanip2}).
Otherwise $P$ and $\pi$ overlap, and by Lemma~\ref{LEM:pcp}, there
are a parallel critical peak
$t' \xL{P'}{s'}\xr{} u'$ and substitutions $\sigma$, $\sigma'$ such
that $\sigma \xR{} \sigma'$ and
$t = t'\sigma' \xL[\Veq\Gamma]{P \setminus P'}
t'\sigma \xL[\Veq\Gamma]{P'}
s'\sigma = s \xr[\Delta]{\epsilon}
u'\sigma = u$ with $P' \subseteq P$.
This case is illustrated in
Figure~\ref{fig-proof-basic}. By assumption there are $u''$, $v$ and $v'$
with $\Var(v|_{Q'}) \subseteq \Var(s|_{P'})$ such that we can join
$t'$ and $u'$ decreasingly, and consequently, using the stability of
labelings we obtain
\[
t'\sigma \xr[\Vee\Gamma]{*}
\cdot \xR[\Veq\Delta]{}
\cdot \xr[\Vee\Gamma\Delta]{*}
v'\sigma \xl[\Vee\Gamma\Delta]{*}
v\sigma \xL[\Veq\Gamma]{Q'}
u''\sigma \xl[\Vee\Delta]{*}
u'\sigma = u
\]
Furthermore, making repeated use of
Lemma~\ref{LEM:pmanip}(\ref{LEM:pmanip2}),
\[
t = t'\sigma' \xr[\Vee\Gamma]{*}
\cdot \xR[\Veq\Delta]{}
\cdot \xr[\Vee\Gamma\Delta]{*}
v'\sigma' \xl[\Vee\Gamma\Delta]{*}
v\sigma' \xL[\Veq\Gamma]{} v\sigma
\]
Notably, the step $v\sigma \xR[\Veq\Gamma]{} v\sigma'$ is obtained
from $s'\sigma \xR[\Veq\Gamma]{} s'\sigma'$
by passing through the rewrite sequence 
$s'\sigma \to u'\sigma \to^* u''\sigma  \xR{} v\sigma$.
We have
$\sigma(x) = \sigma'(x)$ for $x \in \Var(s|_{P'})$ for otherwise
$s \xR[\Gamma]{} t$ would not be a parallel step.
Together with
$\Var(v|_{Q'}) \subseteq \Var(s|_{P'})$, the parallel steps
$u''\sigma \xR[\Veq\Gamma]{} v\sigma$ and
$v\sigma \xR[\Veq\Gamma]{} v\sigma'$ 
can be combined into a single $\xR[\Veq\Gamma]{}$ step by
Lemma~\ref{LEM:pmanip}(\ref{LEM:pmanip2}). Thus we can join $t$ and $u$
decreasingly with common reduct $v'\sigma'$, completing the proof.
\qed
\end{proof}

To conclude the section we demonstrate Theorem~\ref{thm-main-basic} on two
examples. Both are based on rule labeling.

\begin{example}
\label{ex1}
Consider the TRS $\RR$ consisting of the following five rules with labels
$2 > 1 > 0$:
\begin{xalignat*}{5}
\m{a} &\xr[1]{} \m{b} &
\m{b} &\xr[0]{} \m{a} &
\m{f}(\m{a},\m{a}) &\xr[1]{} \m{c} &
\m{f}(\m{b},\m{b}) &\xr[2]{} \m{c} &
\m{h}(x) &\xr[0]{} \m{h}(\m{f}(x,x))
\end{xalignat*}
There are six parallel critical peaks that can all be joined decreasingly
as required by Theorem~\ref{thm-main-basic}:
\begin{xalignat*}{2}
\m{f}(\m{b},\m{a})
&\pcpk{1}{\m{f}(\m{a},\m{a})}{1} \m{c} :&
&\m{f}(\m{b},\m{a})
\xr[\{ 0 \}]{} \m{f}(\m{a},\m{a}) \xr[\{ 1 \}]{} \m{c}
\\
\m{f}(\m{a},\m{b})
&\pcpk{1}{\m{f}(\m{a},\m{a})}{1} \m{c} :&
&\m{f}(\m{a},\m{b})
\xr[\{ 0 \}]{} \m{f}(\m{a},\m{a}) \xr[\{ 1 \}]{} \m{c}
\\
\m{f}(\m{b},\m{b})
&\pcpk{1}{\m{f}(\m{a},\m{a})}{1} \m{c} :&
&\m{f}(\m{b},\m{b})
\xR[\{ 0 \}]{} \m{f}(\m{a},\m{a}) \xr[\{ 1 \}]{} \m{c}
\\
\m{f}(\m{a},\m{b})
&\pcpk{0}{\m{f}(\m{b},\m{b})}{2} \m{c} :&
&\m{f}(\m{a},\m{b})
\xr[\{ 0 \}]{} \m{f}(\m{a},\m{a}) \xr[\{ 1 \}]{} \m{c}
\\
\m{f}(\m{b},\m{a})
&\pcpk{0}{\m{f}(\m{b},\m{b})}{2} \m{c} :&
&\m{f}(\m{b},\m{a})
\xr[\{ 0 \}]{} \m{f}(\m{a},\m{a}) \xr[\{ 1 \}]{} \m{c}
\\
\m{f}(\m{a},\m{a})
&\pcpk{0}{\m{f}(\m{b},\m{b})}{2} \m{c} :&
&\phantom{\m{f}(\m{a},\m{a})}
\mathrel{\phantom{\xR[\{ 0 \}]{}}} \m{f}(\m{a},\m{a}) \xr[\{ 1 \}]{} \m{c}
\end{xalignat*}
Therefore, $\x{R}$ is confluent.
\end{example}

\begin{example}
\label{ex2}
Let $\RR$ be the TRS (Cops~\#62) consisting of the (labeled) rules
\begin{xalignat*}{3}
x - 0 &\xr[0]{} x &
0 - x &\xr[0]{} 0 &
\m{s}(x) - \m{s}(y) &\xr[0]{} x - y \\[-.5ex]
0 < \m{s}(x) &\xr[0]{} \m{true} &
x < 0 &\xr[0]{} \m{false} &
\m{s}(x) < \m{s}(y) &\xr[0]{} x < y \\[-.5ex]
\m{gcd}(x,0) &\xr[0]{} x &
\m{gcd}(0,x) &\xr[0]{} x &
\m{gcd}(x,y) &\xr[1]{} \m{gcd}(y,\m{mod}(x,y)) \\[-.5ex]
\m{if}(\m{true},x,y) &\xr[0]{} x &
\m{if}(\m{false},x,y) &\xr[0]{} y \\[-.5ex]
\m{mod}(x,0) &\xr[0]{} x &
\m{mod}(0,x) &\xr[0]{} 0 \\
\m{mod}(x,\m{s}(y)) &\xr[1]{}
\makebox[0pt][l]{$\m{if}(x < \m{s}(y),x,\m{mod}(x - \m{s}(y),\m{s}(y)))$}
\end{xalignat*}
There are 12 critical pairs, 6 of which are trivial.
One easily verifies that the remaining 6 pairs can be
joined decreasingly, using the order $1 > 0$.
Hence the confluence of $\RR$ follows from Theorem~\ref{thm-main-basic}.
Even though~$\RR$ lacks proper parallel critical pairs, none of the
other results in this paper applies. Note that the preconditions for 
Corollaries~\ref{COR:rtd},~\ref{COR:rt}, and~\ref{COR:red} are not
satisfied as $\RRdnd$, $\ST{\RR}$, and $\RR^\red$ are non-terminating
(due to the rules with label~$1$). Finally, persistence cannot rule out
variable overlaps (of the duplicating $\m{mod}$ rule below the variable~$x$)
and hence Theorems~\ref{THM:pl} and~\ref{THM:pll} based on the rule labeling
fail.
\end{example}

\section{Assessment}
\label{ASS:main}

In this section we relate the results from this article to
each other (Section~\ref{ASS:self})
and to the recent literature~\cite{A10,HM11} (Section~\ref{ASS:other}).

\subsection{Interrelationships}
\label{ASS:self}

\newcommand\Xbig{0.89}
\newcommand\Xsmall{0.89}
\begin{figure}
\subfloat[All.\label{FIG:inter:all}]{
\scalebox{\Xsmall}{
\begin{tikzpicture}[scale=\Xbig,semithick,xshift=-1.2mm]
\draw [color=black] (0,1) circle (1.8);
\draw [color=black] (.866,-.5) circle (1.8);
\draw [color=black] (-.866,-.5) circle (1.8);
\node at (0,0)          {};
\node at (0,-2.9)       {};%
\node at (0,1.9)        {};
\node at (0,-1.3)       {Ex.~\ref{EX:fail:per}};
\node at (1.645,-0.95)  {};
\node at (-1.645,-0.95) {};
\node at (1.126,0.65)   {Ex.~\ref{EX:mot7}};
\node at (-1.126,0.65)  {Ex.~\ref{EX:pll}};
\node at (0,3.2)                    {persistence + rule labeling};
\node at (-2.771,-1.6) [rotate=-60] {parallel rewriting};
\node at (2.771,-1.6) [rotate=60]   {relative termination};
\end{tikzpicture}
}
}
\subfloat[Relative Termination.\label{FIG:inter:rt}]{
\scalebox{\Xsmall}{
\begin{tikzpicture}[scale=\Xbig,semithick]
\draw [color=black] (0,1) circle (1.8);
\draw [color=black] (0,0) circle (0.7);
\draw [color=black] (.866,-.5) circle (1.8);
\draw [color=black] (-.866,-.5) circle (1.8);
\node at (0,0)          {};
\node at (0,-2.9)       {};%
\node at (0,1.9)        {};
\node at (0,-1.3)       {Ex.~\ref{EX:fail:rtd}};
\node at (1.645,-0.95)  {};
\node at (-1.645,-0.95) {};
\node at (1.126,0.65)   {Ex.~\ref{EX:mot4}};
\node at (-1.126,0.65)  {Ex.~\ref{EX:mot5}};
\node at (0,0)                      {Thm.~\ref{THM:l}};
\node at (0,3.2)                    {Corollary~\ref{COR:rtd}};
\node at (-2.771,-1.6) [rotate=-60] {Corollary~\ref{COR:rt}};
\node at (2.771,-1.6) [rotate=60]   {Corollary~\ref{COR:red}};
\end{tikzpicture}
}
}
\caption{Interrelationships.}
\label{FIG:inter}
\end{figure}
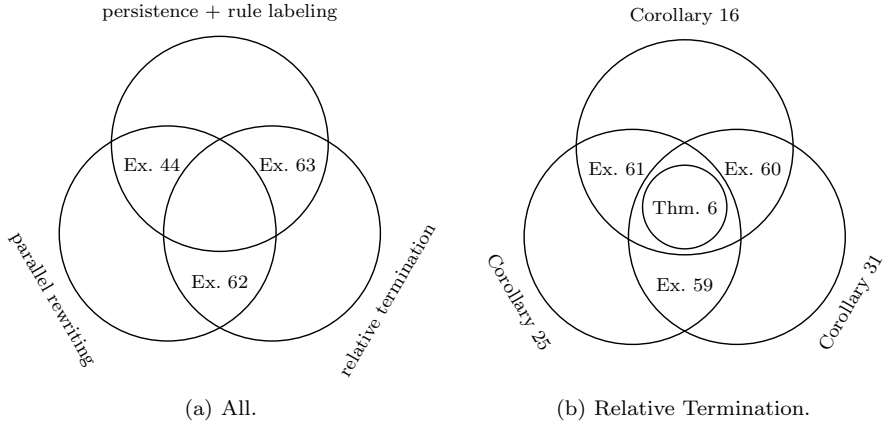

The main results for left-linear systems presented in this article can
be divided into three
classes. Those that require relative termination as a precondition
(Corollaries~\ref{COR:rtd}, \ref{COR:rt}, and~\ref{COR:red}), 
those exploiting persistence (Theorems~\ref{THM:pl} and~\ref{THM:pll}),
and those considering parallel rewriting (Theorem~\ref{thm-main-basic}).
Figure~\ref{FIG:inter}\subref{FIG:inter:all} demonstrates that these three
classes are incomparable.  
The same holds when focusing on the results
relying on relative termination, cf.\
Figure~\ref{FIG:inter}\subref{FIG:inter:rt}.
Note that the regions where only one class is
applicable can be populated with examples using Toyama's celebrated 
modularity result~\cite{T87}, e.g., the disjoint union (after renaming
function symbols) of the TRSs in Examples~\ref{EX:fail:per}
and~\ref{EX:mot7}
can only be handled by the approach based on relative termination.
We discuss the interrelationships in more detail below.

First we observe that Corollaries~\ref{COR:rtd}, \ref{COR:rt},
and~\ref{COR:red} subsume Theorem~\ref{THM:l} since the preconditions of
the corollaries evaporate for linear systems.
The inclusion is strict since Theorem~\ref{THM:l} cannot deal with
the rule $\m{f}(x) \to \m{g}(x,x)$, while all the corollaries can.
Furthermore, Theorem~\ref{THM:l} is subsumed by
Theorem~\ref{THM:pl},
which, if restricted to weak LL-labelings, is subsumed by
Theorem~\ref{THM:pll}.

The following three examples
show that Corollaries~\ref{COR:rtd}, \ref{COR:rt}, and~\ref{COR:red} 
are pairwise incomparable in power (for an overview see 
Figure~\ref{FIG:inter}\subref{FIG:inter:rt}).

\begin{example} 
\label{EX:fail:rtd}
Consider the TRS $\RR$ consisting of the following rules
\begin{xalignat*}{3}
\m{f}(\m{h}(x)) &\to \m{k}(\m{g}(\m{f}(x),x,\m{f}(\m{h}(\m{a})))) &
\m{f}(x) &\to \m{a} &
\m{a} &\to \m{b} \\
\m{k}(x) &\to \m{c} &
\m{b} &\to \bot &
\m{c} &\to \bot
\end{xalignat*}
This TRS has one critical peak (modulo symmetry).
Since $\RRdnd$ is non-terminating, Corollary~\ref{COR:rtd} does not apply.
For Corollary~\ref{COR:rt} observe that $\ST{\RR}$ is terminating using
the interpretation $\m{h_1}_\Nat(x) = x + 1$ and the identify function
for all other function symbols.
To show decreasingness we use the labeling $\ell_\star \times \ell^i_\RL$
with 
${i(\m{f}(x) \to \m{a}) = 1}$
and all other rules receive label $0$.
The critical peak
${t = 
 \m{a} \xl[x,1]{}
 \m{f}(\m{h}(x)) \xr[x,0]{}
 \m{k}(\m{g}(\m{f}(x),x,\m{f}(\m{h}(\m{a})))) = u}$
is closed decreasingly by 
$t 
 \xr[x,0]{} \m{b} 
 \xr[x,0]{} \bot 
 \xl[x,0]{} \m{c}
 \xl[x,0]{} u$.
Corollary~\ref{COR:red} also applies since the polynomial interpretation
with $\m{h}_\Nat(x) = 3x + 1$ and interpreting all other function symbols
by the sum of its arguments establishes termination of $\RR^\red/\RR$.
When taking the identity for $\ell$ in Corollary~\ref{COR:red} the
critical peak
$t = 
 \m{a} \xl[3x+1]{} 
 \m{f}(\m{h}(x)) \xr[3x+1]{} 
 \m{k}(\m{g}(\m{f}(x),x,\m{f}(\m{h}(\m{a})))) 
= u$
can be closed decreasingly by
$t 
 \xr[0]{} \m{b} 
 \xr[0]{} \bot 
 \xl[0]{} \m{c}
 \xl[2x+1]{} u$.
\end{example}

\begin{example}
\label{EX:mot4}
It is easy to adapt the TRS from Example~\ref{EX:mot} such that $\ST{\RR}$
becomes non-terminating. Consider the TRS $\RR$
\begin{xalignat*}{3}
1\colon \m{b} &\to \m{a} &
2\colon \m{a} &\to \m{b} &
3\colon \m{f}(\m{g}(x,\m{a})) &\to \m{g}(\m{f}(x),\m{f}(\m{g}(x,\m{c})))
\end{xalignat*}
for which termination of $\RRdnd$ 
is proved by LPO with precedence $\m{f} > \m{g}$ and
$\m{a} \sim \m{b} > \m{c}$.
Corollary~\ref{COR:rtd} applies since the rule labeling establishes
decreasingness of the critical peak 
$t = \m{f}(\m{g}(x,\m{b})) 
 \xl[2]{} \m{f}(\m{g}(x,\m{a}))
 \xr[3]{} \m{g}(\m{f}(x),\m{f}(\m{g}(x,\m{c}))) = u$
by the join
$t \xr[1]{} \m{f}(\m{g}(x,\m{a})) \xr[3]{} u$. 
Note that 
$\m{f_1}(\m{g_1}(x)) \to \m{g_2}(\m{f_1}(\m{g_1}(x))) \in \STG{\RR}$
is non-terminating and hence Corollary~\ref{COR:rt} does not
apply.%
\footnote{We remark that it is easy to extend this example such that also
$\STT{\RR}$ is non-terminating; just consider the rule
$\m{f}(\m{g}(x,\m{a})) \to
\m{g}(\m{f}(x),\m{g}(\m{f}(\m{g}(x,\m{c}),\m{f}(\m{g}(x,\m{c})))))$.}
For Corollary~\ref{COR:red} the (above) termination proof establishes
termination of~$\RR^\red/\RR$ and $\ell_\red$ in combination with the
rule labeling (taking rule numbers as labels) labels the critical peak 
$t = \m{f}(\m{g}(x,\m{b})) 
 \xl[\m{a},2]{} \m{f}(\m{g}(x,\m{a}))
 \xr[\m{f}(\m{g}(x,\m{a})),3]{} \m{g}(\m{f}(x),\m{f}(\m{g}(x,\m{c}))) = u$
decreasingly since
$t \xr[\m{b},1]{} \m{f}(\m{g}(x,\m{a}))
 \xr[\m{f}(\m{g}(x,\m{a}),3]{} u$. 
\end{example}

\begin{example}
\label{EX:mot5}
Consider the TRS consisting of the rules
\begin{xalignat*}{2}
\m{a}(\m{a}(\m{c})) &\to \m{a}(\m{b}(\m{a}(\m{c}))) &
\m{b}(x) &\to \m{h}(x,x)
\end{xalignat*}
The TRS~$\RR$ has no critical peaks and is terminating by
the following matrix interpretation over $\Nat^2$:
\begin{xalignat*}{2}
\m{a}_{\Nat^2}(\vec x) &=
 \left(\begin{matrix} 1 & 1 \\ 1 & 2 \end{matrix}\right) \vec x +
 \left(\begin{matrix} 0 \\ 3 \end{matrix}\right) &
\m{h}_{\Nat^2}(\vec x, \vec y) &=
 \left(\begin{matrix} 1 & 0 \\ 0 & 0 \end{matrix}\right) \vec x +
 \left(\begin{matrix} 1 & 0 \\ 0 & 0 \end{matrix}\right) \vec y \\
\m{b}_{\Nat^2}(\vec x) &=
 \left(\begin{matrix} 2 & 0 \\ 0 & 0 \end{matrix}\right) \vec x +
 \left(\begin{matrix} 2 \\ 0 \end{matrix}\right) &
\m{c}_{\Nat^2} &=
 \left(\begin{matrix} 0 \\ 0 \end{matrix}\right)
\end{xalignat*}
Hence also $\RRdnd$ is terminating, and by Corollary~\ref{COR:rtd}
the TRS~\RR is confluent.
Corollary~\ref{COR:rt} also applies since $\ST{\RR}$ is terminating.
The derivation $\m{a}(\m{a}(\m{c})) \to \m{a}(\m{b}(\m{a}(\m{c})))
\to_{\RR^\red} \m{a}(\m{a}(\m{c})) \to \cdots$ shows that
$\RR^\red/\RR$ is
non-terminating, so Corollary~\ref{COR:red} does not apply.
\end{example}

Note that any simple monotone reduction pair showing termination of
$\RRdnd$ will also establish termination of $\RR^\red/\RR$,
because if $l \to x \in \RR^\red$ then there is a rule $l \to r \in \RRd$
that duplicates $x$, whence $l > r \geqslant x$.
Hence it is no
surprise that Example~\ref{EX:mot5} used a matrix interpretation of
dimension~2.

\medskip 
Furthermore, the results on relative termination are incomparable with
those on persistence and those based on parallel rewriting.
To this end observe that the first rule of
Example~\ref{EX:pll} violates all preconditions of 
Corollaries~\ref{COR:rtd}, \ref{COR:rt}, and~\ref{COR:red}
but Theorems~\ref{THM:pll} and~\ref{thm-main-basic} apply. 
Note that Theorem~\ref{THM:pll} based on
arbitrary weak LL-labelings subsumes Corollaries~\ref{COR:rtd} 
and~\ref{COR:rt}, since they produce LL-labelings which may be used
to close problematic variable peaks decreasingly even without persistence.
However, if restricted to the rule labeling the following TRS cannot be
handled using persistence while each of the Corollaries~\ref{COR:rtd},
\ref{COR:rt}, and~\ref{COR:red} as well as Theorem~\ref{thm-main-basic} 
succeeds.
\begin{example}
\label{EX:fail:per}
Consider the TRS consisting of the rules
\begin{xalignat*}{2}
1\colon~\m{f}(x,y,\m{a}) &\to \m{f}(x,x,\m{b}) &
2\colon~\m{f}(\m{f}(x,y,\m{b}),z,\m{c}) &\to x
\end{xalignat*}
which is orthogonal. Since a most general sort assignment cannot exclude
variable overlaps of the first rule with itself, Theorem~\ref{THM:pll}
can only succeed when used in combination with an LL-labeling. Note that
all preconditions for Corollaries~\ref{COR:rtd},~\ref{COR:rt},
and~\ref{COR:red}
are satisfied and due to the lack of critical overlaps they are decreasing.
For the same reason Theorem~\ref{thm-main-basic} applies.
\end{example}

The final example shows that Theorem~\ref{thm-main-basic} does not
subsume the plain version for linear TRSs (because of the variable
condition).

\begin{example}
\label{EX:mot7}
Consider the linear TRS consisting of the single rule
\begin{xalignat*}{1}
(x + y) + z &\to (z + y) + x
\end{xalignat*}
Note that all steps are labeled the same, because they use the same rule.
There is only one (parallel) critical peak,
$((z+y)+x)+u \xl{} ((x+y)+z)+u \xr{} (u+z)+(x+y)$, which may
be joined as $((z+y)+x)+u \xr{} ((x+y)+z)+u \xl{} (u+z)+(x+y)$.
Confluence of $\RR$ can be established by Theorem~\ref{THM:l} using
the rule labeling from Lemma~\ref{LEM:lrl}.
On the other hand, trying to use Theorem~\ref{thm-main-basic} fails
for this joining sequence, because
$\Var(((z+y)+x)+u) \not\subseteq \Var((z+y)+x)$. All other ways
of joining the critical peak fail to be decreasing because they
require more than one parallel rewrite step from $((z+y)+x)+u$ or
$(u+z)+(x+y)$, e.g.\ $((z+y)+x)+u \xr{} ((x+y)+z)+y \xr{} (y+z)+(x+y)$.
\end{example}

\subsection{Related work}
\label{ASS:other}

In this section we relate our results to~\cite{HM11,A10}.

To compare our setting with the main result from~\cite{HM11}
we define the \emph{critical pair steps}
$\CPS(\RR) = \{ s \to t, s \to u \mid
\text{$t \cps{s} u$ is a critical peak of $\RR$} \}$.
Furthermore let $\CPS'(\RR)$ be the critical pair steps which do not give
rise to trivial critical pairs.

\begin{theorem}[{\rm\cite[Theorem~3]{HM11}\sf}]
\label{THM:HM11_3}
A left-linear locally confluent TRS $\RR$ is confluent if
$\REL{\CPS'(\RR)}{\RR}$ is terminating.
\end{theorem}

Using the weak LL-labeling $\ell^{\PCPS'(\RR)}_\SN$, from
Theorem~\ref{thm-main-basic} we obtain the following corollary.
Here $\PCPS'(\RR)$ are the parallel critical pair steps which do not give
rise to trivial parallel critical pairs.

\begin{corollary}
\label{COR:pcp}
A left-linear TRS $\RR$  whose parallel critical pairs are joinable
is confluent if
$\REL{\PCPS'(\RR)}{\RR}$ is terminating.
\end{corollary}
\begin{proof}
We need to show that the relative termination assumption eliminates
the variable condition in Theorem~\ref{thm-main-basic}.
If $\REL{\PCPS'(\RR)}{\RR}$ is terminating then for any (non-trivial)
parallel 
critical peak $t \xL[\Gamma]{P} s \xr[\Delta]{} u$ we obtain
$t \xr[\Vee\Gamma]{*} \cdot \xl[\Vee\Delta]{*} u$, hence $Q$ can be chosen
to be empty and $\varnothing = \Var(v|_\varnothing) \subseteq \Var(s|_P)$
trivially holds.%
\footnote{The condition that
$\{s \to t \mid \text{$u \cps{s} t$ is a critical pair}\}/\RR$
is terminating also eliminates the variable condition.}
\qed
\end{proof}

We stress that despite the fact that the preconditions in
Corollary~\ref{COR:pcp}
require more (implementation) effort to check
than those in Theorem~\ref{THM:HM11_3}, 
in theory Corollary~\ref{COR:pcp} subsumes
Theorem~\ref{THM:HM11_3}.
To this end observe that termination of $\PCPS'(\RR)/\RR$ is equivalent
to termination of $\CPS'(\RR)/\RR$. Furthermore 
joinability of the parallel critical pairs is
a necessary condition for confluence just as local confluence is.

Due to the flexibility of the $\ell^\SS_\SN$ labeling we can also choose
$\SS$ to be (a subset of)
the \emph{critical diagram steps}
$\CDS(\RR) = \{ s \to t_i, s \to u_j \mid {}$
$t_0 \cps{s} u_0$ is a critical peak in $\RR$,
$t_0 \xr{*} t_n = u_m \xl{*} u_0$,
$0 \leqslant i \leqslant n$, and
$0 \leqslant j \leqslant m \}$.
Using $\CDS(\RR)$ allows to detect a possible decrease also somewhere in
the joining part of the diagrams.\footnote{%
In~\cite{ZFM11} we employed the strictly weaker system where all steps of
the join (e.g., $t_i \to t_{i+1}$) are used whereas here we use
$s \to t_{i+1}$.}
This incorporates (and generalizes) the idea of critical
valleys~\cite{vO12b}.
However, we remark that our setting does not (yet) follow another recent
trend, i.e., to drop development closed critical pairs
(see \cite{vO12b,HM13}). We leave this for future work.

Next we show that Corollary~\ref{COR:rt} generalizes the results from 
\cite[Sections~5 and~6]{A10}.
It is not difficult to see that the encoding presented
in~\cite[Theorem~5.4]{A10} can be mimicked by Corollary~\ref{COR:rt} where
linear polynomial interpretations over $\Nat$
of the shape as in (1)
\begin{xalignat*}{2}
(1) \quad {f_i}_\Nat(x) &= x + c_f &
(2) \quad {f_i}_\Nat(x) &= x + c_{f_i}
\end{xalignat*}
are used to prove termination of $\ST{\RR}$ and
$\ell_\star \times \ell_\RL$ is employed to show LL-decreasingness of the
critical peaks. In contrast to~\cite[Theorem~5.4]{A10}, which explicitly
encodes these constraints in a single formula of linear arithmetic, our
abstract formulation has the following advantages.
First, we do not restrict
to weight functions but allow powerful machinery for proving relative
termination and second our approach allows to combine arbitrarily many
labelings lexicographically (cf.\ Lemma~\ref{LEM:lllex}). Furthermore we
stress that our abstract treatment of $\ST{\RR}$ allows to implement
Corollary~\ref{COR:rt} based on $\STT{\RR}$
(cf.\ Section~\ref{IMP:main}) which admits further gains in power
(cf.\ Example~\ref{EX:OO03} as well as Section~\ref{EXP:main}).

The idea of the extension presented in~\cite[Example~6.1]{A10} amounts
to using ${\ell_\RL \times \ell_\star}$ instead of
${\ell_\star \times \ell_\RL}$, which is an application of
Lemma~\ref{LEM:lllex} in our setting. Finally, the extension discussed
in~\cite[Example~6.3]{A10} suggests to use linear polynomial
interpretations over $\Nat$
of the shape as in (2)
to prove termination of $\ST{\RR}$. Note that these interpretations are
still weight functions. This explains why the approach
from~\cite{A10} fails to establish confluence of the TRSs in
Examples~\ref{EX:mot} and~\ref{EX:mot2} since a weight function cannot
show termination of the rules
$\m{f_1}(\m{g_1}(x)) \to \m{g_1}(\m{f_1}(x))$ and
$\m{f_1}(\m{h_1}(x)) \to \m{h_1}(\m{g_1}(\m{f_1}(x)))$,
respectively.

Note that both recent approaches~\cite{A10,HM11} based on decreasing
diagrams fail to prove the TRS $\RR$ from Example~\ref{EX:OO03} confluent.
The former can, e.g., not cope with the non-terminating rule
$\m{\times_1}(x) \to \m{+_0}(\m{\times_1}(x))$ in $\STG{\RR}$
(cf.\ Example~\ref{EX:OO03sstar}) while
overlaps with the non-terminating rule ${x + y \to y + x} \in \RR$ prevent
the latter approach from succeeding. 
In contrast, Examples~\ref{EX:OO03d} and~\ref{EX:OO03sstar} give two
confluence proofs based on our setting.

\section{Implementation}
\label{IMP:main}

In this section we sketch how the results from this article can be
implemented.

Before decreasingness of critical peaks can be investigated, the critical
pairs must be shown to be convergent. For a critical pair $t \cp u$ in
our implementation we consider all joining sequences such that 
$t \to^{\leqslant n} \cdot \xl{\leqslant n} u$ and there is no
smaller $n$ that admits a common reduct. 
While in theory longer joining sequences might be easier to label 
decreasingly, preliminary experiments revealed that the effort
due to the consideration of
additional diagrams decreased performance.

To exploit the possibility for incremental confluence proofs by
lexicographically combining labels (cf.\ Lemmata~\ref{LEM:llex}
and~\ref{LEM:lllex}) our implementation 
considers lists of labels.
The search for relative termination proofs (and thus the labelings) is
implemented by encoding the constraints in non-linear (integer) arithmetic.
Below we describe how we combine existing labels (some partial progress)
with
the search for a new labeling to show the critical peaks decreasing.
Note that labelings use different domains (natural numbers, terms), and,
even worse, different orders (matrix interpretations, LPO, etc.). 
The crucial observation for incremental labeling is that neither the
actual labels nor the precise order on the labels have to be recorded but
only how the labels in the join relate to the labels from the peak.
We use the following encoding. Let the local peak have labels
$t \xl[\alpha]{} s \xr[\beta]{} u$. Then a step $v \xr[\gamma]{} w$
is labeled by the pair $(\labvar_\alpha,\labvar_\beta)$ where
$\labvar_\alpha$ and $\labvar_\beta$ indicates if
$\alpha \mathrel{\labvar_\alpha} \gamma$ and
$\beta \mathrel{\labvar_\beta} \gamma$, respectively.
Here $\{\labvar_\alpha,\labvar_\beta\} \subseteq \{\labgt, \labge, \? \}$
and $\?$ means that the labels are incomparable, e.g., 
$\m{f}(x) \? \m{g}(y)$ in LPO or $2x + 1 \? x + 2$ for (matrix)
interpretations.%
\footnote{Our previous implementation (reported in~\cite{ZFM11}) had a bug,
as it did not track incomparable labels properly.}
Decreasingness as depicted in
Figure~\ref{FIG:dec}\subref{FIG:dec:orig} can
then be captured by the conditions shown in
Figure~\ref{FIG:dec}\subref{FIG:dec:enc}, where
$\labany$ can be replaced by any symbol. 
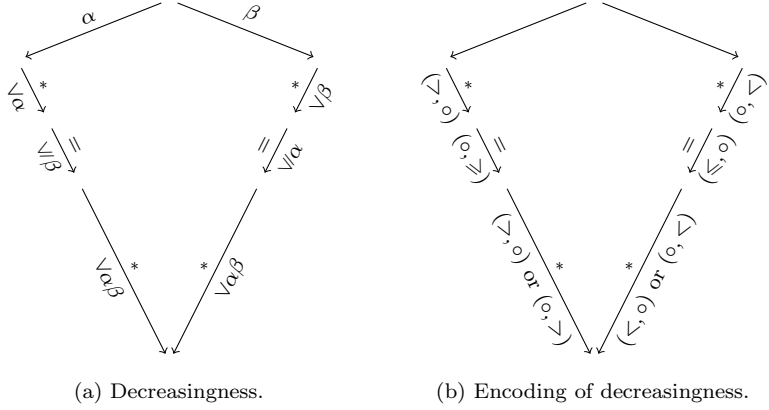
\begin{figure}[t]
\centering
\subfloat[Decreasingness.\label{FIG:dec:orig}]{
\begin{tikzpicture}[xscale=0.8,yscale=0.8]
\node at (0,0)   (1) {};

\node at (-2.5,-1)  (2) {};
\node at ( 2.5,-1)  (3) {};
\node at (-2,-2)  (4) {};
\node at ( 2,-2)  (5) {};

\node at (-1.5,-3)  (6) {};
\node at ( 1.5,-3)  (7) {};
\node at ( 0,-6)   (8) {};

\path[->]  (1) edge node[sloped,above] {$\alpha$} (2); 
\path[->]  (1) edge node[sloped,above] {$\beta$} (3); 

\path[->]  (2) edge node[sloped,below] {$\Vee\alpha$}
                     node[sloped,above] {$\scriptstyle *$} (4); 
\path[->]   (4) edge node[sloped,below] {$\Veq\beta$}
 node[sloped,above] {$=$} (6); 
\path[->]  (6) edge node[sloped,below] {$\Vee\alpha\beta$}
                     node[sloped,above] {$\scriptstyle *$} (8); 

\path[->]  (3) edge node[sloped,below] {$\Vee{\beta}$}
                     node[sloped,above] {$\scriptstyle *$} (5); 
\path[->]   (5) edge node[sloped,below] {$\Veq\alpha$}
 node[sloped,above] {$=$} (7); 
\path[->]  (7) edge node[sloped,below] {$\Vee\alpha\beta$}
                     node[sloped,above] {$\scriptstyle *$} (8); 
\end{tikzpicture}
}
\qquad
\subfloat[Encoding of decreasingness.\label{FIG:dec:enc}]{
\begin{tikzpicture}[xscale=0.8,yscale=0.8]
\node at (0,0)   (1) {};

\node at (-2.5,-1)  (2) {};
\node at ( 2.5,-1)  (3) {};
\node at (-2,-2)  (4) {};
\node at ( 2,-2)  (5) {};

\node at (-1.5,-3)  (6) {};
\node at ( 1.5,-3)  (7) {};
\node at ( 0,-6)   (8) {};

\path[->]  (1) edge node[sloped,above] {} (2); 
\path[->]  (1) edge node[sloped,above] {} (3); 

\path[->]  (2) edge node[sloped,above] {$\scriptstyle *$}
                     node[sloped,below] {$(\labgt,\labany)$} (4); 
\path[->]   (4) edge node[sloped,below] {$(\labany,\labge)$}
 node[sloped,above] {$=$} (6); 
\path[->]  (6) edge node[sloped,above] {$\scriptstyle *$}
                     node[sloped,below]
 {$(\labgt,\labany)\text{ or }(\labany,\labgt)$} (8); 

\path[->]  (3) edge node[sloped,above] {$\scriptstyle *$}
                     node[sloped,below] {$(\labany,\lablt)$} (5); 
\path[->]   (5) edge node[sloped,below] {$(\lable,\labany)$}
 node[sloped,above] {$=$} (7); 
\path[->]  (7) edge node[sloped,above] {$\scriptstyle *$}
                     node[sloped,below]
 {$(\lablt,\labany)\text{ or }(\labany,\lablt)$} (8); 
\end{tikzpicture}
}
\caption{Encoding the order on the labels.}
\label{FIG:dec}
\end{figure}

It is straightforward to implement Corollary~\ref{COR:rtd}. After
establishing termination of $\RRdnd$ (e.g., by an external termination
prover) any weak LL-labeling can be tried to show the critical peaks
decreasing. In \cite{A10,HM11} it is shown how the rule labeling can
be implemented by encoding the constraints in linear arithmetic.
Note that when using weak LL-labelings the implementation does
not have to test condition~\ref{DEF:ll2} in Definition~\ref{DEF:wll}
since this property is intrinsic to weak LL-labelings.

We sketch how to implement the labeling
$\ell^\SS_\SN$ from Lemma~\ref{LEM:lsn} as a relative termination problem.
First we fix a suitable set $\SS$, i.e., the critical diagram steps (see
Section~\ref{ASS:main}).
Facing the relative termination problem $\REL{\SS}{\RR}$ we try to
simplify it according to Theorem~\ref{THM:rt} into some
$\REL{\SS'}{\RR'}$. Note that it is not necessary to finish
the proof. By Theorem~\ref{THM:rt} the relative TRS 
$\REL{(\SS \setminus \SS')}{\RR}$ is terminating and hence
by Lemma~\ref{LEM:lsn} $\ell^{\SS \setminus \SS'}_\SN$ is an
L-labeling. Let ${\geqslant} = {\to^*_\RR}$ and 
${>} = {\to^+_\REL{(\SS \setminus \SS')}{\RR}}$.
Since $\geqslant$ and $>$ can never increase by rewriting, 
it suffices to exploit the first decrease with respect to $>$.
Consider a rewrite sequence
$v_1 \to_\RR v_2 \to_\RR \cdots \to_\RR v_l$.
Take the smallest $k$ such that
$v_1 \to v_{k+1} \in \SS$ but
$v_1 \to v_{k+1} \notin \SS'$.
Then $v_i \to_{(\labge,\labge)} v_{i+1}$ for $1 \leqslant i \leqslant k$
and $v_i \to_{(\labgt,\labgt)} v_{i+1}$ for $k < i < l$.
If no such $k$ exists set 
$v_i \to_{(\labge,\labge)} v_{i+1}$ for $1 \leqslant i < l$. 
We demonstrate the above idea on an example.

\begin{example}
\label{EX:imp}
Consider the following TRS $\RR$ from 
\cite{AYT09}:
\begin{xalignat*}{4}
\m{I}(x) &\to \m{I}(\m{J}(x)) &
\m{J}(x) &\to \m{J}(\m{K}(\m{J}(x))) &
\m{H}(\m{I}(x)) &\to \m{K}(\m{J}(x)) &
\m{J}(x) &\to \m{K}(\m{J}(x))
\end{xalignat*}
We show how the critical peak
$\m{H}(\m{I}(\m{J}(x))) \xl[]{} \m{H}(\m{I}(x))
\to \m{K}(\m{J}(x))$
can be closed decreasingly
$\m{H}(\m{I}(\m{J}(x))) \to_{(\labge,\labge)} 
\m{K}(\m{J}(\m{J}(x))) \to_{(\labgt,\labgt)}
\m{K}(\m{J}(\m{K}(\m{J}(x)))) \xl[(\lable,\lable)]{}
\m{K}(\m{J}(x))$
by $\ell^\SS_\SN$.
Let $\SS$ be the TRS consisting of the critical diagram steps
from the above diagram, i.e., 
\begin{xalignat*}{2}
\m{H}(\m{I}(x)) &\to \m{H}(\m{I}(\m{J}(x))) &
\m{H}(\m{I}(x)) &\to \m{K}(\m{J}(\m{J}(x))) \\
\m{H}(\m{I}(x)) &\to \m{K}(\m{J}(x)) &
\m{H}(\m{I}(x)) &\to \m{K}(\m{J}(\m{K}(\m{J}(x))))
\end{xalignat*}
The interpretation
$\m{H}_\Nat(x) = \m{J}_\Nat(x) = \m{K}_\Nat(x) = x$ and
$\m{I}_\Nat(x) = x+1$ allows to
``simplify'' termination of the problem $\REL{\SS}{\RR}$ according to
Theorem~\ref{THM:rt}. Since the rules that reduce the number of $\m{I}'s$
are dropped from $\SS$ (and $\RR$), those rules admit a decrease in the
labeling.
\end{example}

The abstraction works similarly for the 
labelings~$\ell_\star$ and~$\ell_\red$ from Lemmata~\ref{LEM:star}
and~\ref{LEM:red}, respectively.

Finally, we explain why $\STT{\RR}$ need not be computed explicitly to
implement Corollary~\ref{COR:rt} with the labeling from
Lemma~\ref{LEM:sstar}.
The idea is to start with $\ST{\RR}$ and incrementally prove termination
of $\REL{\STG{\RR}}{\STE{\RR}}$ until some $\REL{\SS_1}{\SS_2}$ is
reached.
If all left-hand sides in $\SS_1$ are distinct then they must have been
derived from different combinations $(l,x)$ with $l \to r \in \RR$ and
$x \in \Var(l)$.\footnote{%
When computing $\ST{\RR}$ the implementation renames 
variables such that $(\ell,x)$ uniquely identifies a rule $\ell \to r$.}
Hence they are exactly those rules which should be
placed in $\STE{\RR}$.
We show the idea by means of an example.

\begin{example}
We revisit Example~\ref{EX:OO03} and try to prove termination
of $\ST{\RR}$. By an application of Theorem~\ref{THM:rt} with the
interpretation given in Example~\ref{EX:OO03sstar} the problem is
termination equivalent to $\REL{\RR_\dagger}{\STE{\RR}}$. By another
application of
Theorem~\ref{THM:rt} the same proof can be used to show termination of
$\REL{(\STG{\RR} \setminus \STD{\RR})}{(\STE{\RR} \cup \STD{\RR})}$
which is a suitable candidate for $\STT{\RR}$ since the rules
in $\STD{\RR}$ have different left-hand sides.
\end{example}

We have also implemented Theorems~\ref{THM:pl} and~\ref{THM:pll}.
The requirements of Theorem~\ref{THM:pl} can be checked effectively
by the following characterization of
$t \in \x{T}_{\trianglelefteq\,\alpha}(\x{F},\x{V})$:

\begin{remark}
\label{REM:testnest}
The condition $t \in \x{T}_{\trianglelefteq\,\alpha}(\x{F},\x{V})$ holds
if and only if $t$ is $S$-sorted and
$S(t) \mathrel{({\leq} \cup {\triangleleft_1})}^* \alpha$, where
the relation $\triangleleft_1$ on sorts relates argument types to
result types: $S(f,i) \triangleleft_1 S(f)$ for all function
symbols $f \in \x{F}$ of arity $n$ and $1 \leqslant i \leqslant n$.
\end{remark}

We only implemented the simplest case of Theorem~\ref{THM:pll},
where $\ell$ is a rule labeling. First, using Remark~\ref{REM:testnest},
we determine for which rules $l \to r \in \x{R}$, $l' \to r' \in \x{R}$,
it is possible to nest $l' \to r'$ below a duplicating variable of
$l \to r$. 
We add constraints $i(l \to r) > i(l' \to r')$ to our constraint
satisfaction problem for the rule labeling. The hard work is done by
an SMT solver.

To postpone the expensive computation (and labeling) of parallel
critical pairs
as long as possible we implemented Theorem~\ref{thm-main-basic} according
the following lazy approach.
We first find ordinary weak LL-labelings for the critical
diagrams, as described earlier in this section.
Only if confluence cannot be established by considering this weak 
LL-labeling for (non-parallel) critical peaks,
we generate parallel critical peaks together
with joining sequences. Finally, we check whether the weak LL-labeling
joins all resulting diagrams (critical and parallel critical) decreasing
as per Theorem~\ref{thm-main-basic}. This check is also responsible for
combining single steps into a parallel
one for the joining sequence.
We confess that this implementation for Theorem~\ref{thm-main-basic}
is somewhat opportunistic but allows to reuse partial progress (the weak
LL-labeling) while postponing parallel critical pairs as long as possible.

\section{Experiments}
\label{EXP:main}

The results from the article have been implemented and form the core
of the confluence prover \CSI~\cite{ZFM11b}. For experiments\footnote{%
\label{FOO:web}%
Details available from
\url{http://cl-informatik.uibk.ac.at/software/csi/labeling2}.}
using version 0.4 of the tool
we considered the current 276 TRSs in Cops. In the experiments we
focus on the 149 systems which have been referenced from the
confluence literature.
From these systems 92 are left-linear.
Our experiments have been performed on a notebook equipped with an
Intel\textsuperscript\textregistered\xspace
quad core processor {i7-2640M} running at a clock rate of 2.8 GHz and
4 GB of main memory.

For 3 systems not even local confluence could be established within
a time limit of 60 seconds.
All other tests finished within this time limit.

\newcommand{\XX}{\phantom{)}}
\newcommand{\mc}[3]{\multicolumn{#1}{#2}{#3}}
\newcommand{\mcc}[1]{\multicolumn{1}{c@{~~}}{#1}}
\begin{table}
\centering
\renewcommand{\arraystretch}{1.25}
\begin{tabular}{@{}l@{\qquad}%
r@{~~}
r@{~~}
r@{~~}
r@{}
}
\hline
method
& \mcc{\PRE} & \mcc{$\CR(\ell_\RL)$} & \mcc{$\CR(\ell_\SN)$} &
\mc{1}{c@{}}{$\CR$}
\\
\hline
Theorem~\ref{THM:l}
& 69
& 42 & 36 & 44
\\
Theorem~\ref{THM:pl}
& 92
& 46 & 40 & 48
\\
Theorem~\ref{THM:pll}
& 92
& 53 & -- & -- 
\\
Corollary~\ref{COR:rtd}
& 65
& 47 & 40 & 49
\\
Corollary~\ref{COR:rt}$\star$
& 66
& 48 & 41 & 50
\\
Corollary~\ref{COR:rt}$\sstar$
& 69
& 51 & 43 & 53
\\
Corollary~\ref{COR:red}
& 65
& 47 & 41 & 49
\\
\hline
Theorem~\ref{thm-main-basic}
& 92
& 55 & 55 & 57
\\
\hline
\end{tabular}
\\[1.0ex]
\caption{Experimental results for 92 left-linear TRSs.}
\label{TAB:trs}
\end{table}

Table~\ref{TAB:trs} shows an evaluation of the results from this article.
The first column indicates which criterion has been used to investigate
confluence. A $\star$ means that the corresponding corollary is
implemented using $\ST{\RR}$ whereas $\sstar$ refers to $\STT{\RR}$.
The column labeled $\PRE$ shows for how many systems the
precondition of the respective criterion is satisfied, e.g., for
Theorem~\ref{THM:l} the precondition is linearity while for
Corollary~\ref{COR:rtd} the precondition is termination of $\RRdnd$.
The columns labeled $\CR(\ell)$ give the number of systems for which
confluence could be established using labeling $\ell$. 
(For Corollary~\ref{COR:rt} implicitly $\ell_\star$ is also employed.
Similarly Corollary~\ref{COR:red} employs $\ell_\red$.)
The column labeled $\CR$ corresponds to the full power of 
each result,
i.e., when the lexicographic combination of all labelings is used. 

From the table we draw the following conclusions. 
On this test bed the labeling function $\ell_\RL$ can handle more systems
than $\ell_\SN$ when considering single steps but 
for parallel rewriting both labelings succeed on equally many systems.
Still, in both settings most power is obtained when using all labelings.
In practice the study of parallel rewriting (Theorem~\ref{thm-main-basic})
is beneficial.
This suggests that the preconditions to obtain weak LL-labelings are
severe.

\begin{table}
\centering
\renewcommand{\arraystretch}{1.25}
\begin{tabular}{
@{}l
r@{~~}
r@{}
}
\\
\hline
tool     &  CR & not CR \\
\hline
\ACP     &  63 & 22 \\
\CSI     &  67 & 20 \\
\SAIGAWA &  53 & 12 \\
\hline
$\sum$   &  68 & 22 \\
\hline
\end{tabular}
\\[1.0ex]
\caption{Comparison with other tools on 92 left-linear TRSs.}
\label{TAB:tools}
\end{table}

For reference in Table~\ref{TAB:tools} we compare the power of the 
confluence provers participating in the Confluence Competition (CoCo),%
\footnote{\url{http://coco.nue.riec.tohoku.ac.jp}}
i.e., \ACP~\cite{AYT09}, \CSI~\cite{ZFM11b}, and
\SAIGAWA~\cite{HM11,KH12}.
\begin{itemize}
\item
\ACP is a powerful confluence prover which implements numerous confluence
criteria from the literature. Its distinctive feature is the strong
support for problems with AC semantics~\cite{AT12}.
\item
\CSI gains most of its power from the labeling framework presented here.
In addition it implements development closed critical pairs~\cite{vO97}
and persistence~\cite{FZM11}.
Recently, the techniques introduced in \cite{AT12} and \cite{KH12} have
also been integrated.
\item
\SAIGAWA also heavily exploits relative termination, remarkably also to
analyze confluence of non-left-linear systems \cite{KH12}.
\end{itemize}

From Tables~\ref{TAB:trs} and~\ref{TAB:tools} we conclude that our
framework
admits a state-of-the-art confluence prover for left-linear systems.

\section{Conclusion}
\label{CON:main}

In this article we studied how the decreasing diagrams technique can be
automated. We presented conditions (subsuming recent related results)
that ensure confluence of a left-linear TRS whenever its critical 
peaks are decreasing.
The labelings we proposed can be combined lexicographically
which allows incremental proofs of confluence and has a modular
flavor in the following sense: Whenever a new labeling function is
invented, the whole framework gains power. We discussed several
situations
(Examples~\ref{EX:OO03}, \ref{EX:mot}, \ref{EX:mot2}, \ref{EX:mot4})
where traditional confluence techniques 
fail but our approach easily establishes confluence.

We have also considered parallel rewriting
resulting in a significantly more powerful approach. We leave the study of
$\XR{}$ and the integration of development closed critical pairs as 
in \cite{vO12b,HM13} as future work.

Recently confluence by decreasing diagrams (for abstract rewrite
systems) has been formalized in the theorem prover
Isabelle/HOL~\cite{Z13,Z13b}. Since the
generated (incremental) labeling proofs are often impossible to check
for humans it seems a natural point for future work to also formalize
the labeling framework to enable automatic certification of confluence
proofs. Since our setting is based on a single method (decreasing
diagrams) while still powerful it offers itself as a perfect candidate
for future certification efforts.

\paragraph{Acknowledgments}
We thank the anonymous reviewers for providing many helpful and detailed
comments.

\providecommand{\noopsort}[1]{}

\end{document}